\documentclass[fleqn,10pt,twocolumn]{SICE_FES26}

\title{Analysis of Dynamic-Key LWE-Based Encrypted Control Systems for Asymptotic Stability and Numerical Safety}

\author{Jungjin Park${}^{\dagger}$ and Kiminao Kogiso}
\speaker{Jungjin Park}

\affils{Department of Mechanical and Intelligent Systems Engineering,\\ Graduate School of Informatics and Engineering,\\
The University of Electro-Communications, Chofu, Tokyo 182-8585,
Japan\\
(E-mail: \{parkjungjin,\,kogiso\}@uec.ac.jp)\\
}
\abstract{%
This study analyzes dynamic-key Learning-with-Errors (LWE)-based encrypted state-feedback control systems with time-varying encoders and decoders.  
Using a Lyapunov-based approach, we derive conditions on the time-varying encoder and decoder parameters that ensure both asymptotic stability and numerical safety by preventing overflow.  
The validity of the derived conditions is confirmed through numerical examples.
}


\usepackage{graphicx} 
\usepackage{subcaption}
\usepackage{mathtools}
\usepackage{amsmath}
\usepackage{amssymb}
\usepackage{amsfonts}
\usepackage{cite}
\usepackage{color}

\newcommand{\KeyGen}{{\mathsf{Gen}}}
\newcommand{\Ecd}{{\mathsf{Ecd}}}
\newcommand{\Dcd}{{\mathsf{Dcd}}}
\newcommand{\Enc}{{\mathsf{Enc}}}
\newcommand{\Encv}{{\mathsf{Enc}}_v}
\newcommand{\EncG}{{\mathsf{Enc}}_G}
\newcommand{\Dec}{{\mathsf{Dec}}}

\newcommand{\sk}{{\mathsf{sk}}}
\newcommand{\Setup}{{\mathsf{Setup}}}

\newcommand{\BitDecomp}{G^{-1}}

\newcommand{\UpdateKey}{{\mathcal{T}}_{{\mathcal{K}}}}
\newcommand{\Updatecv}{{\mathcal{T}}_{{\mathcal{C}}_v}}
\newcommand{\UpdateCG}{{\mathcal{T}}_{{\mathcal{C}}_G}}
\newcommand{\R}{{\mathbb{R}}}
\newcommand{\Q}{{\mathbb{Q}}}
\newcommand{\Z}{{\mathbb{Z}}}
\newcommand{\Zq}{{\mathbb{Z}}_q}
\newcommand{\Zt}{{\mathbb{Z}}_t}
\newcommand{\N}{{\mathbb{N}}}
\newcommand{\M}{{\mathcal{M}}}
\newcommand{\C}{{\mathcal{C}}}

\newcommand{\uniform}{{\mathcal{U}}(\Zq)}
\newcommand{\discrete}{{\mathcal{D}}(0,\sigma)}

\newtheorem{definition}{\textbf{Definition}}

\newtheorem{lemma}{\textbf{Lemma}}

\newtheorem{remark}{\textbf{Remark}}

\newcommand{\textrevise}[1]{\textcolor{black}{#1}}

\begin{document}

\maketitle


\section{Introduction}
\textrevise{Cyber-physical systems (CPSs) integrate physical plants and networked controllers using high-performance computation.
They are expected to enable applications such as collaborative robot systems and smart electric power grids~\cite{Industry5.0survey22}.
However, network connectivity also increases the risk of cyberattacks on measurement and control channels~\cite{Teixerira2015}.
Therefore, mitigating these threats while ensuring secure control systems is essential in CPSs.}

Encrypted control has attracted attention as an effective approach for mitigating these threats~\cite{Kogiso15,Junsoo2016,Darup2018}.
In encrypted control, homomorphic encryption is used to protect controller parameters and signals transmitted over communication networks while still allowing control computations to be performed on encrypted data.
However, the rapid progress of quantum computing raises concerns about the long-term security of conventional public-key cryptosystems, such as ElGamal and Paillier~\cite{Gerjuoy04,Proos03}.
To address this issue, cryptosystems based on the Learning with Errors (LWE) problem have been proposed and are considered resistant to quantum attacks~\cite{Regev05,Peikert09}.
In LWE-based cryptosystems, an intentional encryption error is introduced into the ciphertext to ensure security; however, if this error is not properly managed, it may lead to incorrect decryption.
Several post-quantum cryptosystems (PQC), including the Regev encryption~\cite{Regev05} and the GSW~\cite{GSW}, possess homomorphic properties and have therefore been investigated for encrypted control applications~\cite{Junsoo23,Robert24}.
In addition, a dynamic-key LWE encryption scheme~\cite{Park2026}, in which the private key and ciphertext are jointly updated, has been proposed. 
\textrevise{The additional update operations improve resistance to replay, falsification, and system-identification attacks based on recorded signals, at the cost of increased per-step computational and communication overhead relative to static-key LWE encryption.}

In CPSs, the stability of control systems is crucial, since they should continue to operate correctly under any circumstances. 
Although stability analysis of encrypted control systems based on PQC has been actively studied, it remains a challenging problem.
In~\cite{Junsoo23}, encryption and quantization errors were treated as bounded perturbations, and scaling parameters were proposed to ensure BIBO stability.
In~\cite{Teranishi2024}, the upper bounds on encryption and quantization errors were derived, but it was shown that it is difficult to achieve asymptotic stability with static quantizers alone. 
In~\cite{Ahn2025}, the conditions on time-varying and static quantizers were derived to ensure asymptotic stability under the assumption that encryption errors are negligible. 
However, encryption errors depend on the quantization gain and controller parameters~\cite{Park2026}, and inappropriate selection of quantization gains may result in non-negligible noise. 
Therefore, stability analysis of PQC-based encrypted control systems remains challenging due to the influence of encryption error.

The objective of this study is to derive the conditions of encoders and decoders for reliable encrypted state-feedback control based on a dynamic-key LWE encryption scheme.
The key idea is to evaluate the effects of quantization and encryption errors through a Lyapunov function. 
We derive conditions on the quantization gains, consisting of a static quantizer for the controller and a time-varying quantizer for the state, to ensure asymptotic stability in the presence of quantization and encryption errors.
Moreover, admissible ranges of these gains are characterized to avoid overflow.
Numerical examples demonstrate the effectiveness of the proposed encoders and decoders.
In particular, compared with the method in~\cite{Teranishi20}, which employs static and time-varying encoders without accounting for encryption noise, the proposed approach guarantees asymptotic stability even in the presence of encryption errors.

The contributions of this study are summarized as follows:
(i) We analyze the effects of quantization and encryption errors, and derive conditions on the quantization gains that ensure asymptotic stability of the encrypted state-feedback control system.
(ii)  We characterize the admissible range of quantization gains that guarantees numerical safety, i.e., prevents overflow, in the encrypted state-feedback control system.

\section{Preliminaries}
\subsection{Notations}
The sets of real numbers, rational numbers, integers, natural numbers, plaintext space, and ciphertext space are denoted by $\mathbb{R},\mathbb{Q},\mathbb{Z},\mathbb{N}, \mathcal{M},\mathcal{C}$, respectively.
The rounding is denoted by $\left\lceil\cdot\right\rfloor$.
We define ${\mathbb{Z}}^+ := \{z \in {\mathbb{Z}} \; | \; 0 \leq z\}$, ${\mathbb{Z}}_q := \{z \in {\mathbb{Z}} \; | \; -\frac{q}{2} < z \leq \frac{q}{2}\}$, and ${\mathbb{R}}^+ := \{r \in {\mathbb{R}} \; | \; 0 \leq r\}$, respectively.
For a scalar $a \in \R$, its absolute value is denoted by $|a|$.
The set of vectors of size $n$ is denoted by $\R^n$. 
The $j$th element of a vector $v$ is denoted by $v_j$.
The norm $\ell_2$ and the infinite norm of $v$ are denoted by $||v||$ and $\|v\|_\infty$, respectively.
The set of matrices of size $m \times n$ is denoted by $\R^{m \times n}$.
The $(i,j)$ entry of matrix $M$ is denoted by $M_{ij}$.
The minimum and maximum eigenvalue of $M$ are denoted by $\lambda_{\min}(M)$ and $\lambda_{\max}(M)$, respectively.
For $r\in\N$, $I_r\in\R^{r\times r}$ denotes the identity matrix.
The set of positive-definite symmetric matrices is denoted by ${\mathbb{S}}_{+}$.
The uniform distribution over $\Zq$ is denoted by $\uniform$, and the zero-mean discrete Gaussian distribution with standard deviation $\sigma>0$ is denoted by $\discrete$.
The gadget matrix $G$ is defined as
$G:=\begin{bmatrix} I_r & \nu I_r & \cdots & \nu^{d-1} I_r\end{bmatrix}$, where $\nu$ is the radix, and $d\in \mathbb{N}$ is chosen such that $\nu^{d-1}<q\le\nu^d$.
For an arbitrary vector $w=[w_0,\ldots,w_{r-1}]^{\top}\in\Zq^{r}$, the digit decomposition map $\BitDecomp(\cdot)$ expands each entry of $w$ into its base-$\nu$ representation and is defined as $\BitDecomp(w):=[w_{0,0},\ldots, w_{0,d-1},\,w_{1,0}, \ldots,w_{r-1,d-1}]^{\top}\in\Z^{dr}$,
where $w_i=\sum_{j=0}^{d-1} w_{i,j}\nu^j$ with $w_{i,j}\in\{-\nu+1,\ldots,\nu-1\}$. 
With these definitions, the relation $w=G\,\BitDecomp(w) $ holds.

\subsection{Dynamic-Key LWE Encryption Scheme}
From~\cite{Park2026}, the key-updatable LWE-based encryption scheme is presented as follows.
A dynamic-key LWE-based  encryption scheme $\Pi$ at step $k$ is defined as the tuple,
\begin{align*}
\Pi
:= (\Setup, \KeyGen, \Enc_v, \Enc_G, \Dec, \UpdateKey, \Updatecv, \UpdateCG),
\end{align*}
where $(\Setup, \KeyGen, \Enc_v, \Enc_G, \Dec)$ are the algorithms of a conventional LWE-based  encryption scheme~\cite{Junsoo20},
and the remaining tuple ($\UpdateKey$, $\Updatecv$, $\UpdateCG$) are update transition maps.
The algorithms are described as follows.
 \begin{itemize}
    \item $\Setup(1^\lambda)$: Choose a dimension $n\in\N$,
    modulus $t=2^{t_0}$, $t_0\in\N$, modulus $q=2^{q_0}$, $q_0\in\N$, such that $t\le q$, and a standard deviation $\sigma>0$.
    Choose $d\in\N$ and define $\nu:=2^{\nu_0}$ such that $\nu^{d-1}<q\le \nu^d$. Define the plaintext space $\M:=\Zt$, and the ciphertext spaces $\C_v:=\Zq^{n+1}$, $\C_G:=\Zq^{(n+1)\times d(n+1)}$. 
    Return ${\mathsf{p}}=(n,t,q,\sigma,\nu,d)$.

    \item $\KeyGen({\mathsf{p}})$: Generate the private key $\sk\in\Zq^n$ as a column vector, with each element sampled from $\uniform$.
    Return $\sk$ and $\tau = \begin{bmatrix}1\\ \sk\end{bmatrix}$.

    \item $\Encv(m\in\Zt,\sk)$: Given input $m\in\Zt$ and $\sk$, return the ciphertext $c\in\Zq^{n+1}$, 
    where $a\in\Zq^n$ is a random column vector and $\epsilon \in\Z$ is an error sampled from $\discrete$,  $b=-\sk^{\top} a+\tfrac{q}{t}m+ \epsilon \bmod q$, and $c=\begin{bmatrix}b\\ a\end{bmatrix}$.
    For simplicity, the ciphertext of a vector $X$ encrypted component-wise by $\Encv$ is denoted by $c_X$.
    
    \item $\EncG(m\in\Zt,\sk)$: Given input $m\in\Zt$ and $\sk$, return the ciphertext $C\in\Zq^{(n+1)\times d(n+1)}$,
    where $A\in\Zq^{n\times d(n+1)}$ is a random matrix, and $E\in\Zq^{d(n+1)}$ is an error row vector, with each element sampled from $\discrete$, $B=-\sk^{\top}A+E \bmod q$, and $C=mG+\begin{bmatrix}B\\ A\end{bmatrix}\bmod q$.
    For simplicity, the ciphertext of a vector $X$ encrypted component-wise by $\EncG$ is denoted by $C_X$.

    \item $\Dec(c\in\Zq^{n+1},\sk)$: Given ciphertext $c$ and private key $\sk$, compute $\left\lceil\tfrac{t}{q}\tau^\top c\right\rfloor = \left\lceil m + \tfrac{t}{q}e \right\rfloor $. If $\tfrac{t}{q}|e|<\tfrac{1}{2}$, then $\Dec(c,\sk)=m$.
    In general, let $\sk(k)$ denote the private key at step $k$ and $c(k) = \Encv(m,\sk(k))$.
    For $m\in\Zt$, if the parameters $t$ and $q$ satisfy $|\tfrac{t}{q} e(k)| < \tfrac{1}{2}$, then $ \Dec(c(k),\sk(k)) = m$, where $e(k)$ is the error generated by $\Encv$.
   Defining the encryption error as $\delta :=\Dec(c(k),\sk(k))-m$, if the inequality is satisfied, $\delta=0$; otherwise, $\delta\neq 0$. 
    Since $m \in \Zt$, the rounding operation yields $\left\lceil m + \tfrac{t}{q} e \right\rfloor = m + \left\lceil \tfrac{t}{q} e \right\rfloor$.
    Thus, the encryption error can be expressed as $\delta =\left\lceil \tfrac{t}{q} e \right\rfloor$.
     Additionally, the decryption algorithm $\Dec$ is not defined for matrix ciphertext $C$, since only a vector ciphertext $c$ is decrypted in encrypted control systems.

    \item $\UpdateKey$: The private key update map is defined as
    \begin{align*}
      &\UpdateKey:\ (\sk(k), \tau(k)) \mapsto (\sk(k+1), \tau(k+1)) \\
      &\quad = \left(\sk(k)+s(k)\bmod q, \begin{bmatrix}1\\ \sk(k)+s(k)\end{bmatrix}\bmod q\right),
    \end{align*}
    where \textrevise{$s(k)\in\Zq^n$} is a random column vector with each element sampled from $\uniform$.

    \item $\Updatecv$: The ciphertext update map for $c$ is defined as
    \begin{align*}
      \Updatecv:\ c(k) \mapsto c(k+1)
      =\begin{bmatrix}
        b-s(k)^{\top}a\\
        a
      \end{bmatrix}\bmod q,
    \end{align*}
    where $c(k)$ is the ciphertext encrypted by $\Encv$ and $\sk(k)$, and $s(k)\in\Zq^n$ is a random column vector with each element sampled from $\uniform$ at step $k$.

    \item $\UpdateCG$: The ciphertext update map for $C$ is defined as
    \begin{align*}
      \UpdateCG: C(k)\mapsto &C(k+1)  \\
      &  =mG+\begin{bmatrix}
        B-s(k)^{\top}A\\
        A
      \end{bmatrix}\bmod q,
    \end{align*}
    where $C(k)$ is the ciphertext encrypted by $\EncG$ and $\sk(k)$, and $s(k)\in\Zq^n$ is a random column vector with each element sampled from $\uniform$ at step $k$.
\end{itemize}

The homomorphic operations are defined as follows.
For ciphertexts $c_1(k), c_2(k) \in \Zq^{n+1}$, define the homomorphic addition as $c_1 \oplus c_2 := c_1(k) + c_2(k)$.
Moreover, for ciphertext $C \in \Zq^{(n+1)\times d(n+1)}$ and $c \in \Zq^{n+1}$, define the homomorphic multiplication as $C(k) \otimes c(k) := C(k) \BitDecomp(c(k))$.
In LWE-based encryption schemes, homomorphic operations increase the noise introduced during encryption. 
Specifically, let $\sk(k)$ be the private key at step $k$, and let $c_1(k) = \Encv(m_1,\sk(k))$ and $c_2(k) = \Encv(m_2,\sk(k))$ be ciphertexts with errors $\epsilon_1(k)$ and $\epsilon_2(k)$ generated by $\Encv$, respectively.
 For $m_1$ and $m_2$, $\Dec(c_1(k) \oplus c_2(k), \sk(k)) = m_1 + m_2 + \left\lceil \tfrac{t}{q}\epsilon^{{\mathsf{add}}} \right \rfloor$, where $\epsilon^{{\mathsf{add}}}(k) := \epsilon_1(k) + \epsilon_2(k)$.
Similarly,
let $\sk(k)$ be the private key at step $k$, $C(k) = \EncG(m_1, \sk(k))$, and $c(k) = \Encv(m_2, \sk(k))$, with errors $E(k)$ and $\omega(k)$ generated by $\EncG$ and $\Encv$, respectively.  
For $m_1, m_2 \in \Zt$, $\Dec(C(k) \otimes c(k), \sk(k)) = m_1 m_2+ \left\lceil \tfrac{t}{q}\epsilon^{{\mathsf{mult}}} \right \rfloor$, where $\epsilon^{\mathsf{mult}}(k) := m_1 \omega (k) + E(k) \BitDecomp(c(k))$.

\section{Problem Description}
This section describes the encoder/decoder design problem for reliable encrypted control systems.

\subsection{Unencrypted Control Systems}
We consider discrete-time state-feedback control systems with a linear plant given by:
\begin{align}
x(k+1)=A_px(k)+B_pu(k),\quad u(k)=Fx(k), \label{eq:system}
\end{align}
where $k\in\Z^+$, $x\in\R^{n_p}$, $u\in\R$, and $F\in\R^{n_p}$, denote the step, plant state, control input, and feedback gain, respectively. 
It assumes that the state is measurable, the pair $(A_p, B_p)$ is controllable, and the gain $F$ is designed such that $A_p+B_pF=:A_c$ is Schur stable.

\subsection{Encoder and Decoder}
Since the state and feedback gains are real-valued, they must be quantized to construct an encrypted control system. 
In this study, the quantization process is realized using the encoder $\Ecd_{\gamma}$ and decoder $\Dcd_{\gamma}$, defined as
\begin{align*}
\Ecd_{\gamma}\!: &\,\R\ni x\mapsto\check{x}=\left\lceil\gamma x\right\rfloor\in\M,\\
\Dcd_{\gamma}\!: &\,\M\ni\check{x}\mapsto\bar{x}=\gamma^{-1}\check{x}\in\Q,
\end{align*}
where $\gamma$ is a time-varying quantization gain, which plays an important role in the properties of the encrypted controls. 
The measured state $x(k)$ is encoded at every step using a quantization gain $\gamma(k)=\gamma_1(k)$, whereas the gain $F$ is encoded once using a static quantization gain $\gamma(k)=\gamma_0$ prior to the implementation of the control system.
Specifically, encoding the gain is written as 
\begin{align}
\Ecd_{\gamma_0}(F) &= \left\lceil \gamma_0 F\right\rfloor, \label{ecd0}
\end{align}
and encoding and decoding the signals are respectively 
\begin{align}
&\Ecd_{\gamma_1(k)}(x(k))=\left\lceil \gamma_1(k) x(k)\right\rfloor,\label{ecd1}\\
&\Dcd_{\gamma_2(k)}(\check{F}\check{x}(k))=\frac{1}{\gamma_2(k)}\check{F}\check{x}(k),\label{dcd2}
\end{align}
where $\gamma_2:=\gamma_0\gamma_1$.
Additionally, both $\Ecd_{\gamma}$ and $\Dcd_{\gamma}$ are applied element-wise to vectors.

The encoded values must lie within the plaintext space \textrevise{$\M=\Zt$} to avoid overflow during control operations.
Each element $m \in \M$ satisfies $-\frac{t}{2}<m\leq\frac{t}{2}$ from the definition of $\Zt$.
If the encoded value exceeds this range, an overflow occurs due to the modulo-$t$ operation, leading to incorrect decoding. 
Therefore, the encoder must be designed by appropriately choosing the quantization gain to avoid overflow.
To formalize this requirement, the overflow (underflow) is defined as follows.

\begin{definition}\label{def:overflow}
Let $\M=\Zt$ and $\check{x}=\Ecd_{\gamma}(x)$. 
If there exists an index $j\in\{1,\ldots,n_p\}$ such that $|\check{x}_j|\ge\frac{t}{2}$ holds, then an overflow is said to occur.
\end{definition}

\subsection{Controller Encryption}
The state-feedback control system \eqref{eq:system} is reconstructed using a dynamic-key LWE-based encryption scheme to obtain the encrypted control system, as follows.
Let $C_F(k)$, $c_x(k)$, and $c_u(k)$ denote the ciphertexts of the gain $F$, the state $x(k)$, and the control input $u(k)$ at step $k$, respectively.
Then, the encrypted controller, defined as $f_{\Pi}\!:(C_F(k),c_x(k))\mapsto c_u(k)$, is given by
\begin{align}\label{eq:enc_ctrl}
&(\sk(k+1),\tau(k+1)) = \UpdateKey(\sk(k),\tau(k)),\notag\\
&C_F(k+1) = \UpdateCG(C_F(k)),\notag\\
&c_u(k)=\bigoplus_{j=1}^{n_p}C_{F_j}(k)\otimes c_{x_j}(k),
\end{align}
where ${\mathsf{p}}=\Setup(1^{\lambda})$, $\sk(0)=\KeyGen({\mathsf{p}})$, $C_F(0)=\EncG(\check{F},\sk(0))$ with $\check{F}=\Ecd_{\gamma_0}(F)$, and $c_x(k)=\Encv(\Ecd_{\gamma_1(k)}(x(k)),\sk(k))$.
The encrypted controllers perform computations on the encrypted data, such as $F$, $x(k)$, and $u(k)$, using a time-varying (updatable) private key $\sk(k)$ to maintain security over time.

\textrevise{
\begin{remark}
We consider an honest-server setting in which the server processes only encrypted data.
For ciphertext updates, the random matrix~$A$ in $C_F(k)$ and the update vector~$s(k)$ must be stored securely on the server, since their exposure may compromise the encrypted control system.
The same~$s(k)$ is generated at the plant and the server by sharing a pseudo-random number generator seed and synchronizing the step~$k$.
\end{remark}
}

\subsection{Problem Statement}
The resulting LWE-based encrypted control system involves two types of errors: the quantization-induced and encryption-induced errors, denoted by $\delta_x:=\Dcd(\Ecd(x))-x=\bar{x}-x\in\R^{n_p}$ and for given $c_u$, $\delta_u:=\Dcd(\Dec(c_u))-\Dcd(\check{F}\check{x})\in\R$, respectively.
The former arises from converting real-valued signals into plaintext elements, while the latter is inherent to the encryption and \textrevise{decryption} processes.
These errors directly affect the control performance and may lead to instability or numerical issues such as overflow.
Therefore, it is essential to characterize their impact and account for them in the design of the encoder and decoder.
Based on this observation, we introduce the notion of reliability for encrypted control systems as follows.

\begin{definition}
An encrypted control system is said to be reliable if it satisfies the following two conditions:
(i) Stability: The resulting control system achieves asymptotic stability in the presence of errors, and
(ii) Numerical safety: No overflow occurs during operation.
\end{definition}

We consider the following problem.
\begin{problem}
Clarify the conditions on the quantization gains $\gamma_0$ and $\gamma_1$ in \eqref{ecd0}--\eqref{dcd2} to ensure that the resulting encrypted control system is reliable.
\end{problem}
We address this problem by analyzing the effects of quantization and encryption errors, deriving the corresponding conditions.

Figs.~\ref{fig:config} show the configurations of the original (unencrypted), quantized, and encrypted control systems. 
The quantized control system in Fig.~\ref{fig:config}(b) is a control system, in which the quantization error $\delta_x$ is added to the input through the encoder and decoder. 
Moreover, the encrypted control system in Fig.~\ref{fig:config}(c) further includes the encryption error $\delta_u$, which is added to the input in addition to the quantization error. 
Consequently, an encrypted control system based on the LWE encryption scheme is affected by both quantization and encryption errors, which potentially destabilize the system. 
This observation motivates the analysis carried out in this study to ensure reliable encrypted control system.

\begin{figure}[tb]
    \centering
    \hspace{2ex}
    \vspace{3ex}
    \begin{subfigure}{0.43\linewidth}
        \centering
        \includegraphics[width=\linewidth]{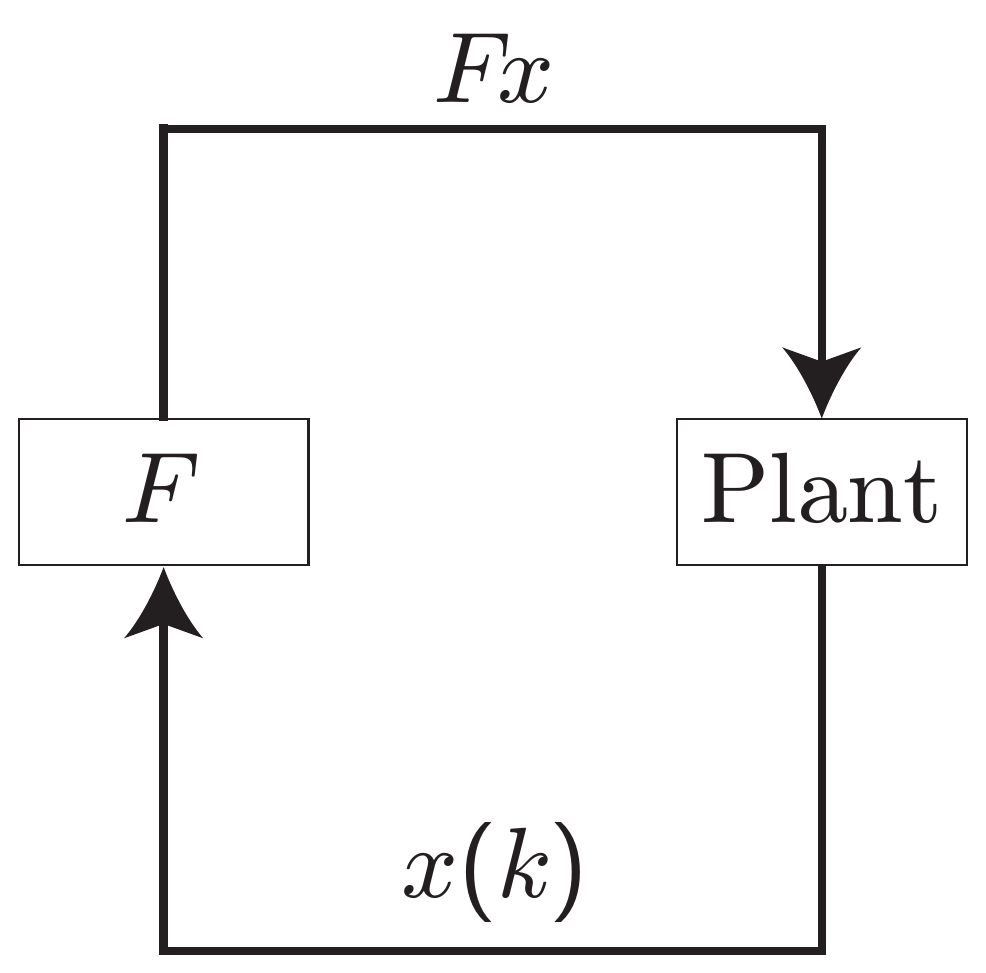}
        \caption{Original control system}
        \label{fig:org}
    \end{subfigure}
    \hfill
    \begin{subfigure}{0.48\linewidth}
        \centering
        \includegraphics[width=\linewidth]{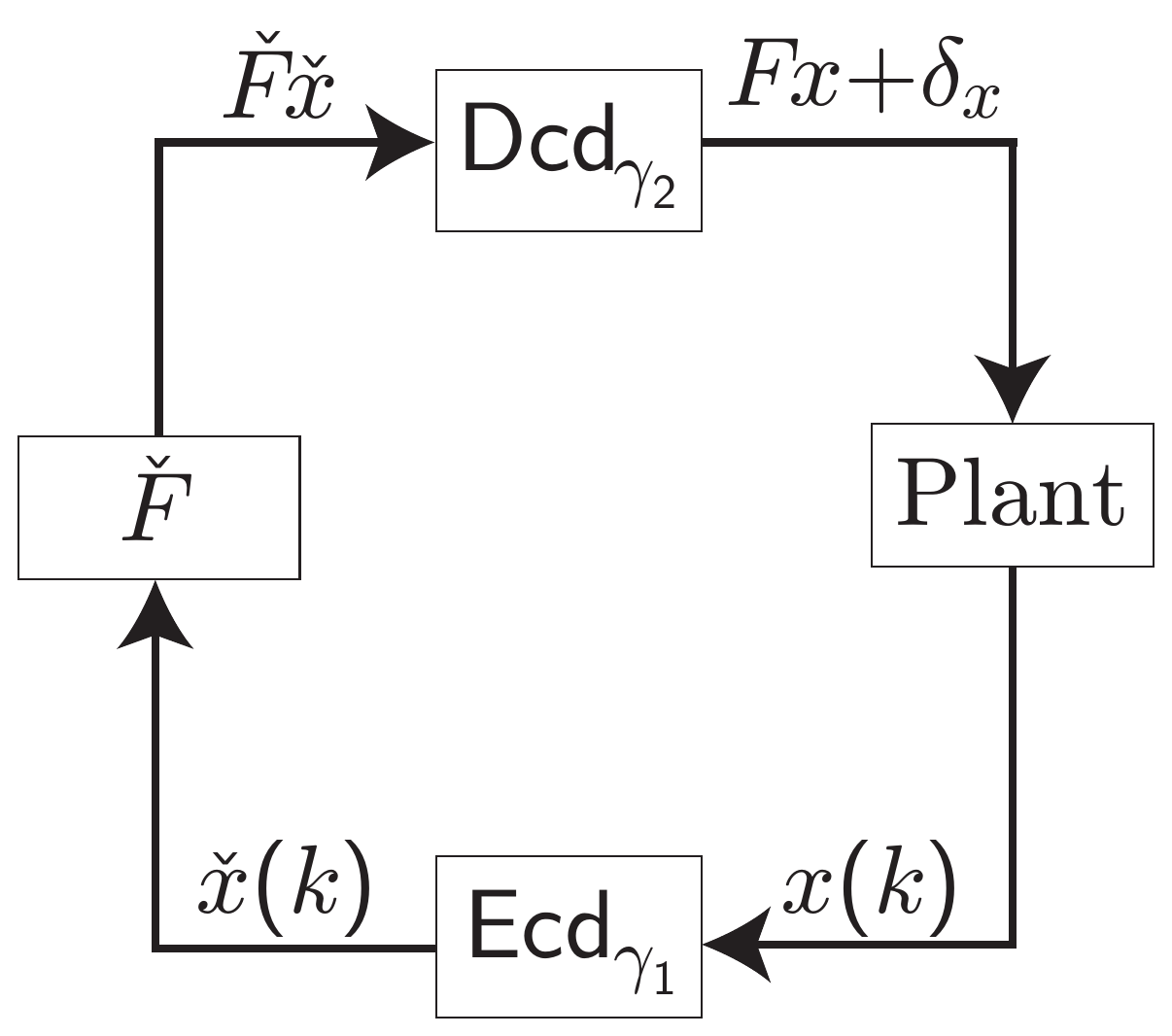}
        \caption{Quantized control system}
        \label{fig:quant}
    \end{subfigure}
    \begin{subfigure}{0.90\linewidth}
        \centering
        \includegraphics[width=\linewidth]{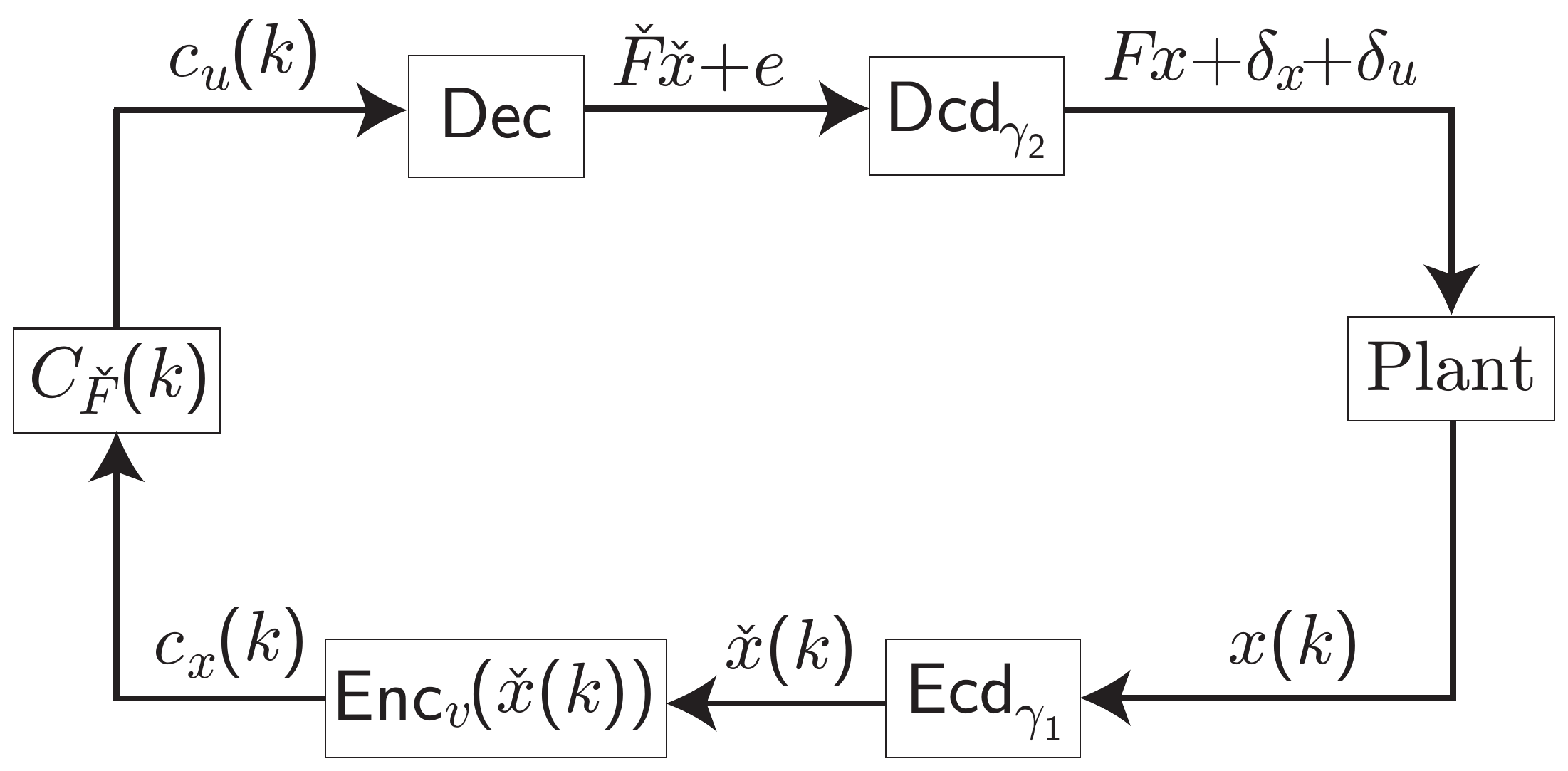}
        \caption{Encrypted control system}
        \label{fig:enc}
    \end{subfigure}
    \caption{Configurations of original, quantized, and encrypted control systems.}
    \label{fig:config}
    \vspace{-3ex}
\end{figure}

\section{Error Bound Analysis}\label{sec:ea}
This section characterizes the bounds of the quantization and encryption.
The quantization error satisfies the following bound~\cite{Teranishi20}.
\begin{lemma} \label{lem:quanterror}
Under $\gamma_0=\gamma_1\equiv\bar{\gamma}$, the quantization error of the state $x\in \R^{n_p}$ satisfies  $\|\delta_x\|\leq\frac{\sqrt{n_p}}{2}|\bar{\gamma}|^{-1}$.
\end{lemma}

The encryption error is characterized as follows.
\begin{lemma}\label{lem:delta}
For given parameters ${\mathsf{p}}$, $\delta_u$ is bounded as
\begin{align}\label{ieq:decerror}
    |\delta_u| \leq \frac{1}{\gamma_2}\left\lceil\frac{t}{q}n_p\kappa \sigma (\|\Ecd_{\gamma_0}(F)\|_\infty + (n+1)d\nu)\right\rfloor,
\end{align}
where $\kappa\in\mathbb{N}$ is a constant chosen such that the tail probability 
$\Pr(|\omega|>\kappa\sigma)$ is negligible.
\end{lemma}
\begin{proof}
Using \eqref{eq:enc_ctrl}, the encryption error is given by $\delta_u=\Dcd(\Dec(c_u))-\Dcd(\check{F}\check{x})$, where the first term can be expressed as
\begin{align*}
&\Dcd(\Dec(c_u))
=\Dcd\left(\Dec\left(\bigoplus_{j=1}^{n_p}C_{F_{j}}\otimes c_{x_j}\right)\right),\\
&=\Dcd\left(\left\lceil \frac{t}{q}\tau^\top \bigoplus_{j=1}^{n_p}C_{F_{j}}\otimes c_{x_j})\right\rfloor\right),\\
&=\gamma_2^{-1}\left \lceil \frac{t}{q}\left(\sum_{j=1}^{n_p} \left(\frac{q}{t} \check{F}_j\check{x}_j + (\check{F}_j \omega+E \BitDecomp(c_{x_j}))\right)\right)\right\rfloor,\\
&=\gamma_2^{-1}\left \lceil \sum_{j=1}^{n_p} \left(\check{F}_j\check{x}_j+\frac{t}{q} (\check{F}_j\omega+E \BitDecomp(c_{x_j}))\right)\right\rfloor,\\
&=\gamma_2^{-1}\left\lceil  \check{F}\check{x} + \tfrac{t}{q} e \right\rfloor=\gamma_2^{-1}\left(\check{F}\check{x}+\left\lceil \tfrac{t}{q} e\right\rfloor\right), \ (\because \check{F}\check{x}\in\mathbb{Z})
\end{align*}
where $e=\sum_{j=1}^{n_p} \check{F}_j \omega+E\BitDecomp(c_{x_j})$, and the second term is given by $\Dcd(\check{F}\check{x})=\gamma_2^{-1}\check{F}\check{x}$.
Thus, the error becomes  $\delta_u=\gamma_2^{-1}\lceil\tfrac{t}{q}e\rfloor$.
Next, we clarify the bounds of the error term $e$.
The first term of $e$ satisfies  $|\check{F}_j\omega| \le \|\check{F}_j\|_\infty |\omega|\le \|\check{F}_j\|_\infty\kappa\sigma$ since $|\omega|<\kappa\sigma$.
The second term of $e$ can be bounded as $|E\BitDecomp(c_{x_j}(k))|\le (n+1)d\nu \kappa\sigma$ since $|E_i|\le \kappa\sigma$ and $\BitDecomp(c_{x_j}(k))$ has $(n+1)d$ elements.
Summing over the $n_p$ multiplications yields $|e|\leq n_p (\|\Ecd_{\gamma_0}(F)\|_\infty+(n+1)d\nu)\kappa\sigma$, which leads to \eqref{ieq:decerror}.
\end{proof}

These bounds will be used to analyze the effect of the errors on the encrypted control system. 

\begin{remark}
In quantization error analysis, the maximum width between plaintext elements must be considered. 
For ElGamal encryption, the plaintext space is a cyclic group, so this width must be explicitly evaluated. 
In contrast, LWE-based schemes use $\Zt$, where plaintext elements are uniformly spaced with unit width, i.e., the distance between adjacent elements is 1.
\end{remark}

\section{Condition for Reliability}
This section analyzes the quantization errors arising from the static encoder for the controller gain and the time-varying encoder for the system state, and derives admissible ranges of the quantization gains that ensure asymptotic stability of the closed-loop system.
Moreover, this section presents the admissible range of quantization gains for which overflow and underflow do not occur.

\begin{figure*}[t]
\begin{align}
&h_0(P,Q):=\|B_p^\top P B_p\|^{-1}\left(\sqrt{\| A_c^\top PB_p\|^2+\lambda_{\min}(Q)\| B_p^\top PB_p\|}-\|A_c^\top P B_p\|\right),\label{h0}\\
&h_1(\bar{P},\bar{Q}):={\lambda^{-1}_{\min}(\bar{Q})}\left( \frac{\sqrt{n_p}}{2}\|\bar{A}_c^{\top}\bar{P}B_p\bar{F}\|+  \right.\notag \\
&\left.\sqrt{\left(\frac{\sqrt{n_p}}{2}\|\bar{A}_c^\top\bar{P}B_p\bar{F}\|+\bar{e}\|\bar{A}_c^\top\bar{P}B_p\|\right)^2 + \lambda_{\min}(\bar{Q})\left(\frac{n_p}{4}\|\bar{F}^\top B_p^\top\bar{P}B_p\bar{F}\|+\frac{\bar{e}^2}{\gamma_0^2}\|B_p^\top\bar{P}B_p\|+\frac{\sqrt{n_p}\bar{e}}{\gamma_0}\|B_p^\top \bar{P}B_p\bar{F}\|\right)} \right),\label{h1}\\
&\bar{e}({\mathsf{p}},\gamma_0,F):=\left|\left\lceil\frac{t}{q}n_p\kappa\sigma(\|\Ecd_{\gamma_0}(F)\|_\infty+(n+1)d\nu)\right\rfloor\right|, \label{eq:maxe}
\end{align}
\vspace{1mm}
\noindent\rule{\textwidth}{0.3pt}
\vspace{-2mm}
\end{figure*}

\subsection{Main Result}
The main result of this study is presented in the following theorem, which clarifies the condition under which the encrypted control system is reliable.

\begin{theorem}\label{thm:results}
If, for given $\mathsf{p}$ of $\Pi$, the quantization gains $\gamma_0$ and $\gamma_1$ in \eqref{ecd0}--\eqref{dcd2} satisfy the following conditions, respectively,
\begin{subequations}\label{thm:gamma}
\begin{align}
&\frac{2}{\sqrt{n_p}\, h_0(P,Q)}<|\gamma_0|<\frac{t-1}{2\|F\|_\infty}, \label{thm:gamma0}\\
&\frac{h_1(\bar{P},\bar{Q})}{\|x(k)\|}<|\gamma_1(k)|<\frac{t-1}{2n_p\|\check{F}\|_{\infty}\|x(k)\|_{\infty}},\ \ \forall k\in\mathbb{Z}^+, \label{thm:gamma1}
\end{align}
\end{subequations}
then the resulting LWE-based encrypted control system corresponding to \eqref{eq:system} is reliable. 
Here, $h_0(P,Q)$, $h_1(\bar{P},\bar{Q})$, and $\bar{e}({\mathsf{p},\gamma_0,F})$ are defined by \eqref{h0}--\eqref{eq:maxe}, and
$(P,Q)$ and $(\bar{P},\bar{Q})\in\mathbb{S}_{+}\times\mathbb{S}_{+}$ are the solutions to $A_c^{\top}PA_c-P=-Q$ and  $\bar{A}_c^{\top}\bar{P}\bar{A}_c-\bar{P}=-\bar{Q}$, respectively, where $\bar{A}_c:=A_p+B_p\bar{F}$.
\end{theorem}

\subsection{Proof of Theorem 1}
The overview of the proof is as follows.
A Lyapunov approach is used to analyze the stability of the state-feedback control system in the presence of quantization and encryption errors. 
The analysis results in lower bounds on the gains $\gamma_0$ and $\gamma_1$ to ensure asymptotic stability, as well as upper bounds on these gains to prevent overflow and guarantee numerical safety. 
Consequently, it is shown that the quantization gains satisfying \textbf{Theorem~\ref{thm:results}} achieve a reliable encrypted control system.

\subsubsection{Quantization error on gain}
The following lemma provides a condition on $\gamma_0$ under which the feedback-gain-quantized control system remains Schur stable~\cite{Teranishi20}.
\begin{lemma}\label{lemma:gamma0}
Consider a state-feedback control system with an encoded gain, i.e., $x(k+1)=A_px(k)+B_pu(k)$ and $u(k)=\bar{F}x(k)$, \textrevise{where $\check{F}=\Ecd_{\gamma_0}(F)$ and $\bar{F}=\Dcd_{\gamma_0}(\check{F})$.}
If $\gamma_0$ satisfies~\eqref{thm:gamma0}, then $\bar{A}_c\,(=A_p+B_p\bar{F})$ is Schur stable.
\end{lemma}

In addition to ensuring Schur, the quantization gain $\gamma_0$ must be appropriately selected to avoid overflow and underflow.
The following lemma gives a sufficient condition.
\begin{lemma}\label{lem:gammac}
\textrevise{If $\gamma_0$ satisfies}
\begin{align}\label{ieq:gamma0}
|\gamma_0|< \frac{t-1}{2||F ||_{\infty}},  
\end{align}
then overflow of $\check{F}$ does not occur.
\end{lemma}
\begin{proof}
Assume that the condition~\eqref{ieq:gamma0} holds. 
Then, for each element $F_j$ of $F$ with $j \in \{1,\ldots,n_p\}$, we have
$|\gamma_0 F_j| \leq |\gamma_0| \|F\|_{\infty} < \frac{t-1}{2}$.
Since $\check{F}_j = \left\lceil \gamma_0 F_j \right\rfloor$, the rounding function satisfies
$\left| \left\lceil \gamma_0 F_j \right\rfloor \right| \leq \left \lceil|\gamma_0| \|F\|_{\infty} \right\rfloor \leq |\gamma_0| \|F\|_{\infty} + \frac{1}{2} < \frac{t-1}{2} + \frac{1}{2} = \frac{t}{2}$.
Therefore, if the condition~\eqref{ieq:gamma0} is satisfied, it guarantees that $|\check{F}_j| < \frac{t}{2}$ for all $j$, and hence no overflow occurs by \textbf{Definition~\ref{def:overflow}}.
\end{proof}

\subsubsection{Quantization error on communication signals}
Under the assumption that $\bar{A}_c$ is Schur, we analyze the effect of time-varying quantization on the communication signals and encryption error.
The following lemma characterizes a condition on the time-varying quantization gain $\gamma_1(k)$ under which the closed-loop system achieves asymptotic stability.
\begin{lemma}\label{theorem:2}
Consider a state-feedback control system with an encoded gain and communication signals, i.e., $x(k+1)=A_px(k)+B_pu(k)$ and $u(k)=\bar{F}\bar{x}(k)+\delta_u$, \textrevise{where $\check{F}=\Ecd_{\gamma_0}(F)$, $\check{x}=\Ecd_{\gamma_1}(x)$, $\bar{F}=\Dcd_{\gamma_0}(\check{F})$, and $\bar{x}=\Dcd_{\gamma_1}(\check{x})$.}
Assume that $\bar{A}_c$ is Schur stable.
If there exist $\gamma_0$ and $\gamma_{1}(k)$, $\forall k\in\mathbb{Z}^+$, such that the conditions \eqref{thm:gamma} are satisfied, then the resulting quantized control system is asymptotically stable.
\end{lemma}
\begin{proof}
Using $\bar{x}=\delta_x+x$, the quantized control system is written by 
\begin{align} \label{eq:effect_error}
x(k+1)
&=A_px(k)+B_p(\bar{F}\bar{x}(k)+\delta_u(k)),\nonumber\\
&=\bar{A}_cx(k)+B_p\bar{F}\delta_x(k) + B_p\delta_u(k).
\end{align}
Now that $\bar{A}_c$ is Schur, there exists $\bar{P}\in{\mathbb{S}}_{+}$ for any $\bar{Q}\in{\mathbb{S}}_{+}$ such that $\bar{P}=\bar{A}_c^{\top}\bar{P}\bar{A}_c+\bar{Q}$ holds.
Let $V_k(x):=x^{\top}(k)\bar{P}x(k)$ be a Lyapunov function candidate, then
\begin{align}
&V_{k+1}(x)-V_k(x) \notag\\
&=x^\top\bar{A}_c^\top \bar{P}B_p\bar{F}\delta_x+x^\top\bar{A}_c^\top\bar{P}B_p \delta_u\notag\\ 
&\quad+\delta_x\bar{F}^\top B_p^\top\bar{P}\bar{A}_c\,x+\delta_x^\top\bar{F}^\top B_p^\top\bar{P}B_p\bar{F}\delta_x\notag\\
&\quad+\delta_x^\top \bar{F}^\top B_p^\top\bar{P} B_p\delta_u+\delta_uB_p^\top \bar{P}\bar{A}_c x\notag\\
&\quad+\delta_uB_p^\top \bar{P}B_p\bar{F}\delta_x+\delta_uB_p^\top\bar{P}B_p\delta_u-x^\top\bar{Q}x,\notag\\
&\leq-\lambda_{\min}(\bar{Q})\|x\|^2+2\|\bar{A}_c^{\top}\bar{P}B_p\bar{F}\|\|\delta_x\|\|x\|\notag\\
&\quad+\|\bar{F}^{\top}B_p^{\top}\bar{P}B_p\bar{F}\|\|\delta_x\|^{2}+2\|\bar{A}_c^\top\bar{P}B_p\|\|\delta_u\|\|x\|\notag\\
&\quad+\|B_p^\top \bar{P} B_p\| \|\delta_u\|^2+2\|B_p^\top \bar{P} B_p\bar{F}\| \|\delta_x\|\|\delta_u\|\label{eq:quadra}.
\end{align}
Consider the solution $\eta$ of the quadratic equation in $\|x\|$ given in~\eqref{eq:quadra}. 
Based on \textbf{Lemmas~\ref{lem:quanterror}} and~\textbf{\ref{lem:delta}}, $\|\delta_x\|$ and $\|\delta_u\|$ are bounded as \textrevise{$\|\delta_x\|\le\frac{\sqrt{n_p}}{2\gamma_1(k)}$}, and $\|\delta_u\|=|\delta_u|\leq\frac{1}{\gamma_0\gamma_1(k)}\bar{e}$, where $\bar{e}$ is defined as~\eqref{eq:maxe}, respectively.
Substituting these bounds into the solutions of the quadratic equation yields the upper bound, i.e., $\eta \leq \gamma_1(k)^{-1}h_1(\bar{P},\bar{Q})$, where $h_1$ is defined as~\eqref{h1}.
Moreover, since $V_{k+1}(x) - V_k(x)$ is concave, $||x|| > \gamma_1(k)^{-1}h_1(\bar{P},\bar{Q})$ implies that $V_{k+1}(x)-V_k(x)$ becomes negative.
Therefore, if $\gamma_1$ satisfies~\eqref{thm:gamma1}, then $V_{k+1}(x)-V_k(x)<0$, which implies that the closed-loop system is Lyapunov stable.

Moreover, the Lyapunov function satisfies $0\leq\lambda_{\min}(\bar{P}) \|x_k\|^2\le V_k(x)\le\lambda_{\max}(\bar{P}) \|x_k\|^2$, where $0\leq\lambda_{\min}(\bar{P})$ since $\bar{P}\in{\mathbb{S}}_+$.
Summing the inequality $V_{k+1}(x) - V_k(x) \le -{\mathsf{c}}\|x_k\|^2$
from $k=0$ to $N-1$ yields $V_N(x) - V_0(x)\le -{\mathsf{c}}\sum_{k=0}^{N-1} \|x_k\|^2$ where ${\mathsf{c}}\in \R^+$ and $N\in \N$.
Since $V_N(x) \ge 0$, it follows that ${\mathsf{c}}\sum_{k=0}^{\infty} \|x_k\|^2\le\lim_{N\to\infty} (V_0(x) - V_N(x))\le V_0(x)$, which implies that $\lim_{k\to\infty}\|x_k\|=0 $.
Therefore,  the closed-loop system is asymptotically stable.
\end{proof}

\begin{remark}
It follows from the condition $\|x\| > \gamma_1(k)^{-1} h_1(\bar{P}, \bar{Q})$ that as the state norm $\|x\|$ decreases, the quantization gain $\gamma_1$ must increase to satisfy the inequality.
In particular, if $x \to 0$ as $k \to \infty$, then $\gamma_1(k) \to \infty$.
This implies that the quantization gain increases over time, which supports the asymptotic stability of the closed-loop system.
\end{remark}

Moreover, the quantization gain $\gamma_1$ must be appropriately selected to prevent overflow and underflow. 
The following lemma gives a sufficient condition.

\begin{lemma}\label{lem:gammap}
Suppose $\gamma_0$ satisfies~\eqref{ieq:gamma0}.
If $\gamma_1(k)$ satisfies
 \begin{align}\label{ieq:gamma1}
|\gamma_1(k)|< \frac{t-1}{2n_p|\gamma_0|||F||_{\infty}||x(k)||_{\infty}} 
 \end{align}
then overflow of $\check{F}\check{x}$ does not occur.
\end{lemma}
\begin{proof}
Assume that the condition~\eqref{ieq:gamma1} holds and $\gamma_0$ satisfies~\eqref{ieq:gamma0}.
The encoded input $\check{F}\check{x}$ is given by $\check{F}\check{x}(k) = \sum_{j=1}^{n_p} \check{F}_j \check{x}_j(k)$ with $j \in \{1,\ldots,n_p\}$.
Then, for encoded input, we have,
$\sum_{j=1}^{n_p} |\gamma_0 F_j| |\gamma_1(k)x(k)| \leq n_p|\gamma_0| \|F\|_{\infty} |\gamma_1(k)|||x(k)||_{\infty} < \frac{t-1}{2}$.
Since $\check{F}_j = \left\lceil \gamma_0 F_j \right\rfloor$ and $\check{x}_j(k) = \left\lceil \gamma_1(k) x_j(k) \right\rfloor$, the rounding function satisfies
$\sum_{j=1}^{n_p}\left| \left\lceil \gamma_0 F_j \right\rfloor \right|\left| \left\lceil \gamma_1(k) x_j(k) \right\rfloor \right| 
\leq n_p|\gamma_0| \|F\|_{\infty}|\gamma_1(k)| \|x(k)\|_{\infty} + \frac{1}{2} < \frac{t-1}{2} + \frac{1}{2} = \frac{t}{2}$, which completes this proof.
\end{proof}

Therefore, from \textbf{Lemmas~\ref{lemma:gamma0}} and \textbf{\ref{theorem:2}}, the stability condition of reliability is guaranteed. 
Moreover, from \textbf{Lemmas~\ref{lem:gammac}} and \textbf{\ref{lem:gammap}}, overflow is prevented, thereby ensuring the numerical safety condition of reliability.
Consequently, if the quantization gains $\gamma_0$ and $\gamma_1$ in \eqref{ecd0}--\eqref{dcd2} satisfy the conditions in \textbf{Lemmas~\ref{lemma:gamma0}} to \textbf{\ref{lem:gammap}}, then the reliability of the encrypted control system is guaranteed.
This completes the proof of \textbf{Theorem~\ref{thm:results}}.

\begin{remark}
This study analyzes the stability condition, taking into account the encryption error $\delta_u$, which differs from that in~\cite{Teranishi20}.
Thus, \textbf{Lemmas~\ref{lemma:gamma0}} and \textbf{\ref{lem:gammac}} are given by~\cite{Teranishi20}, while \textbf{Lemmas~\ref{theorem:2}} and \textbf{\ref{lem:gammap}} are original contributions of this study.
\end{remark}

\textrevise{
\begin{remark}
The above analysis also applies to static-key LWE-based encrypted control~\cite{Junsoo20}, because the encryption error bound~\eqref{ieq:decerror} does not depend on the step $k$.
\end{remark}}

\begin{remark} 
From~\cite{Kawase2022}, it is suggested that $\gamma_0$ should be chosen as small as possible within a range that ensures the stability of the encoded gain, while $\gamma_1$ should be taken as large as possible to improve quantization accuracy. 
This choice is consistent with empirical design guidelines in dynamic quantization. 
On the other hand, as in~\cite{Azuma2008} and~\cite{Teranishi2021}, more systematic approaches based on the optimization of encoders with respect to appropriately defined objective functions have also been proposed. 
The investigation of such optimal design methods for $\gamma_0$ and update rule $\gamma_1$ is left for future work.
\end{remark}

\section{Numerical Example}
This section demonstrates the design of the encoder and decoder for encrypted control through numerical simulations.
The computations were performed on a MacBook Air equipped with an Apple M4 chip and 24\,GB of memory, running macOS Tahoe 26.5.
The encryption scheme was implemented using the C++ NTL library.

\subsection{Example Setup}
Consider a state-feedback control system~\eqref{eq:system} with
\begin{align*}
    &A_p = \begin{bmatrix} 1.0000 & 0.0100 & 0.0000\\ -0.0001 & 0.9999 & 0.0099 \\ -0.0298 & -0.0200 & 0.9900 \end{bmatrix},\notag\\
    &B_p = \begin{bmatrix} 0.0100 \\ 0.0100 \\ -0.0002 \end{bmatrix},\notag
\end{align*}
and $F = \begin{bmatrix}-1.8249 & -1.2118 & 0.3565\end{bmatrix}$, which was designed by the LQR using input and state weights set to $1$ and $I_3$, respectively.
The parameters $\kappa$, $Q$, and $\bar{Q}$ were set to $6$, $I_3$, and $I_3$, respectively.
The encryption parameters were chosen as ${\mathsf{p}}=(n,t,q,\sigma,\nu,d) = (7,2^{40},2^{40},4.0,2,40)$. 
These parameters achieve a security level of $\lambda_Q = 128$ bits according to the Lattice estimator~\cite{LWE_estimator} with the number of ciphertexts set to $\mathsf{m}=7$.
The matrix $P$ was computed as
\begin{align*}
    P = 
    \begin{bmatrix}
    29.5329 & -25.3125 & -5.0626 \\
    -25.3125 & 82.4275 & 2.1743 \\
    -5.0626 & 2.1743 & 61.3562
    \end{bmatrix}.
\end{align*}
The quantization gain was selected as $\gamma_0 = 1.1622$, and the resulting encoded gain and quantized gain are given by $\check{F} = \begin{bmatrix} -2 & -1 & 0 \end{bmatrix}$, and $\bar{F} = \begin{bmatrix} -1.7208 & -0.8604 & 0 \end{bmatrix}$, respectively.
$\bar{e} = 92736$, and the matrix $\bar{P}$ was computed as
\begin{align*}
    \bar{P} = 
    \begin{bmatrix}
    27.6351 & -19.1927 & -15.6086 \\
    -19.1927 & 95.0851 & -1.8341 \\
    -15.6086 & -1.8341 & 101.0115
    \end{bmatrix}.
\end{align*}
The time-varying quantization gain was chosen as $\gamma_1(0) = 2.878 \times 10^{4}$ with the initial state $x(0) = \begin{bmatrix}1 & 1 & 1\end{bmatrix}^\top$.
\textrevise{At each step, $\gamma_1(k)$ is updated from the measured $\|x(k)\|$ as
$\gamma_1(k)=\min\!\;\{(t-1)/(2n_p\|\check{F}\|_\infty\|x(k)\|_\infty),\; h_1(\bar{P},\bar{Q})/\|x(k)\|+0.01\}$ to satisfy Theorem~\ref{thm:results} with bias~$0.01$.}

\subsection{Results}
Fig.~\ref{fig:dyn} shows the results of the encrypted control using the static and time-varying encoders based on \textbf{Theorem~\ref{thm:results}}, as well as those using only static encoders with $\gamma_1(k) \equiv 2.878 \times 10^{4}$. 
Fig.~\ref{fig:dyn_wrong} shows the corresponding results obtained using the method in~\cite{Teranishi20}.

Fig.~\ref{fig:dyn}(a) and Fig.~\ref{fig:dyn_wrong}(a) show the time response of $\delta_u$. 
In Fig.~\ref{fig:dyn}, the red and blue lines correspond to the response obtained with the time-varying encoders designed according to \textbf{Theorem~\ref{thm:results}} and the static encoders, respectively. 
In Fig.~\ref{fig:dyn_wrong}, the blue line represents the response obtained with the time-varying encoders from~\cite{Teranishi20}.
Figs.~\ref{fig:dyn}(b)--(d) and \ref{fig:dyn_wrong}(b)--(d) show the time responses of the states $x_1$, $x_2$, and $x_3$, respectively. 
In Fig.~\ref{fig:dyn}, the red and blue lines correspond to the responses with the time-varying and static encoders, respectively, while the gray line denotes the response of the original system. 
In Fig.~\ref{fig:dyn_wrong}, the blue line represents the responses with the time-varying encoders from~\cite{Teranishi20}, and the gray line indicates the responses of the original system.
Fig.~\ref{fig:dyn}(e) and Fig.~\ref{fig:dyn_wrong}(e) show the time response of the input $u$. 
In Fig.~\ref{fig:dyn_wrong}, the blue line represents the response with the time-varying encoders from~\cite{Teranishi20}, and the gray line indicates the response of the original system.
Finally, Fig.~\ref{fig:dyn}(f) and Fig.~\ref{fig:dyn_wrong}(f) show the behavior of $\gamma_1$. 
In Fig.~\ref{fig:dyn}, the red and blue lines show the behavior of the quantization gain $\gamma_1$ for the proposed time-varying and static encoders, respectively, and the red dashed line represents the upper bound to prevent overflow. 
In Fig.~\ref{fig:dyn_wrong}, the blue line shows the behavior of $\gamma_1$ for the time-varying encoders from~\cite{Teranishi20}.
\textrevise{
However, static encoders show residual errors in the states and input.
Fig.~\ref{fig:dyn} shows input residuals of order $10^{-3}$ and state residuals of order $10^{-4}$.
Once $\|x(k)\|$ becomes sufficiently small that $\gamma_1 < h_1(\bar{P},\bar{Q})\|x(k)\|^{-1}$, the fixed~$\gamma_1$ violates
the lower bound in~\eqref{thm:gamma1}, and the static encoders fail to achieve asymptotic stability.}
Moreover, Fig.~\ref{fig:dyn_wrong} shows that the quantization gain without noise evaluation fails to properly handle encryption noise, resulting in significant deviations from the original control signals.
These results demonstrate that the proposed encoding and decoding scheme based on \textbf{Theorem~\ref{thm:results}} enables reliable encrypted control systems.

\begin{figure}[tb]
    \centering
    \begin{subfigure}{0.47\linewidth}\label{fig:diffu}
        \centering
        \includegraphics[width=\linewidth]{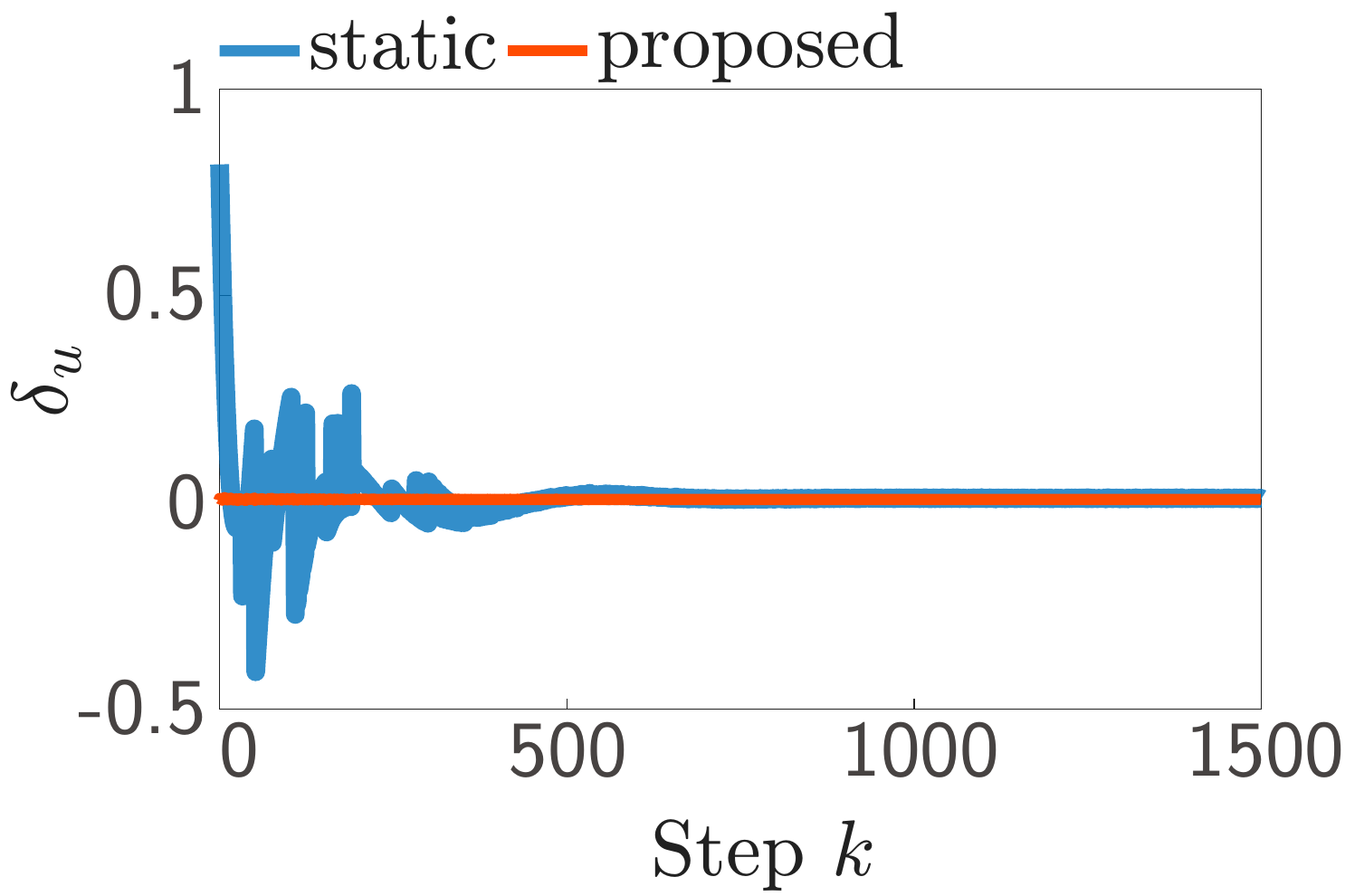}
        \caption{$\delta_u$}
    \end{subfigure}
    \hfill
    \begin{subfigure}{0.47\linewidth}
        \centering
        \includegraphics[width=\linewidth]{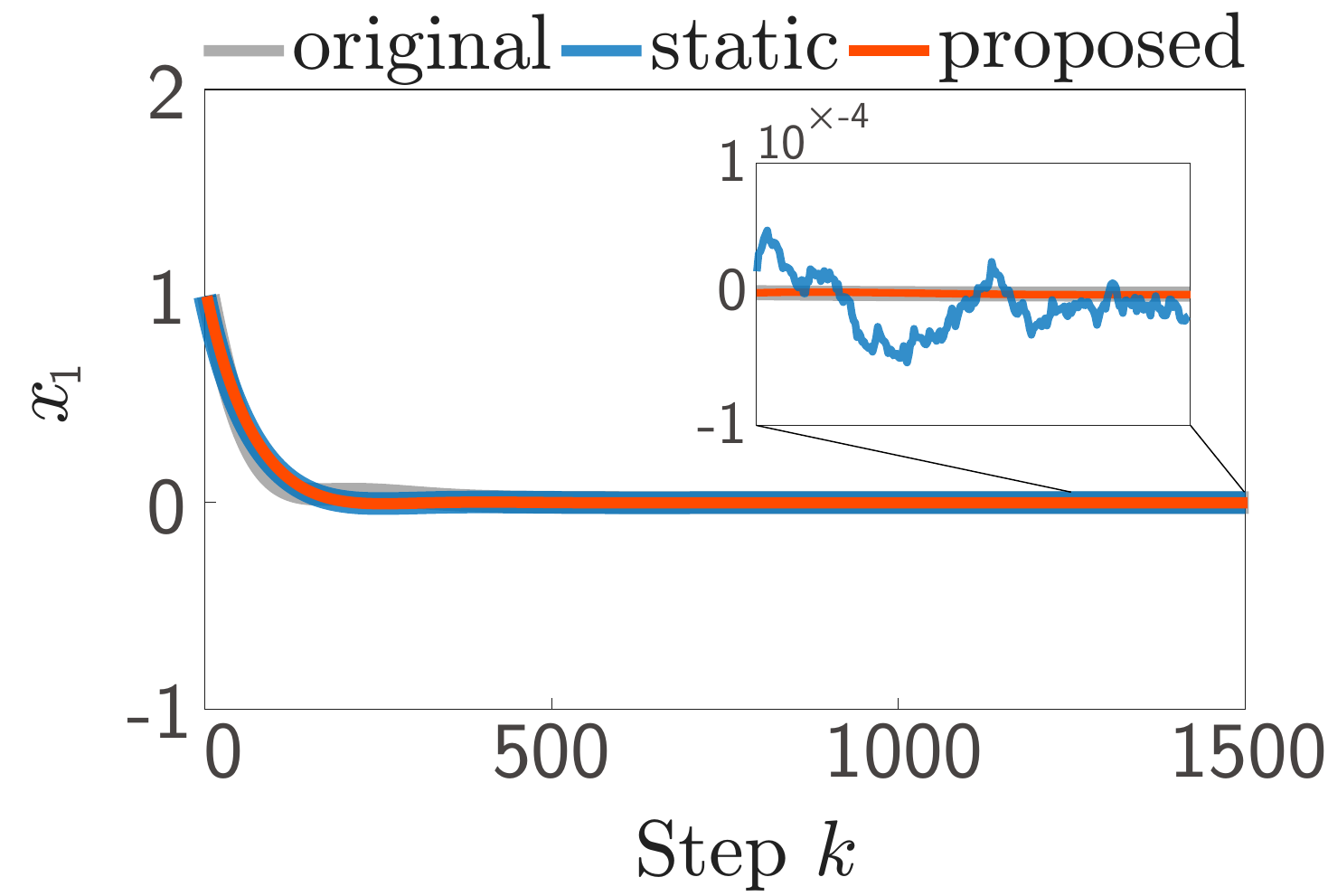}
        \caption{$x_1$}
    \end{subfigure}

    \begin{subfigure}{0.47\linewidth}
        \centering
        \includegraphics[width=\linewidth]{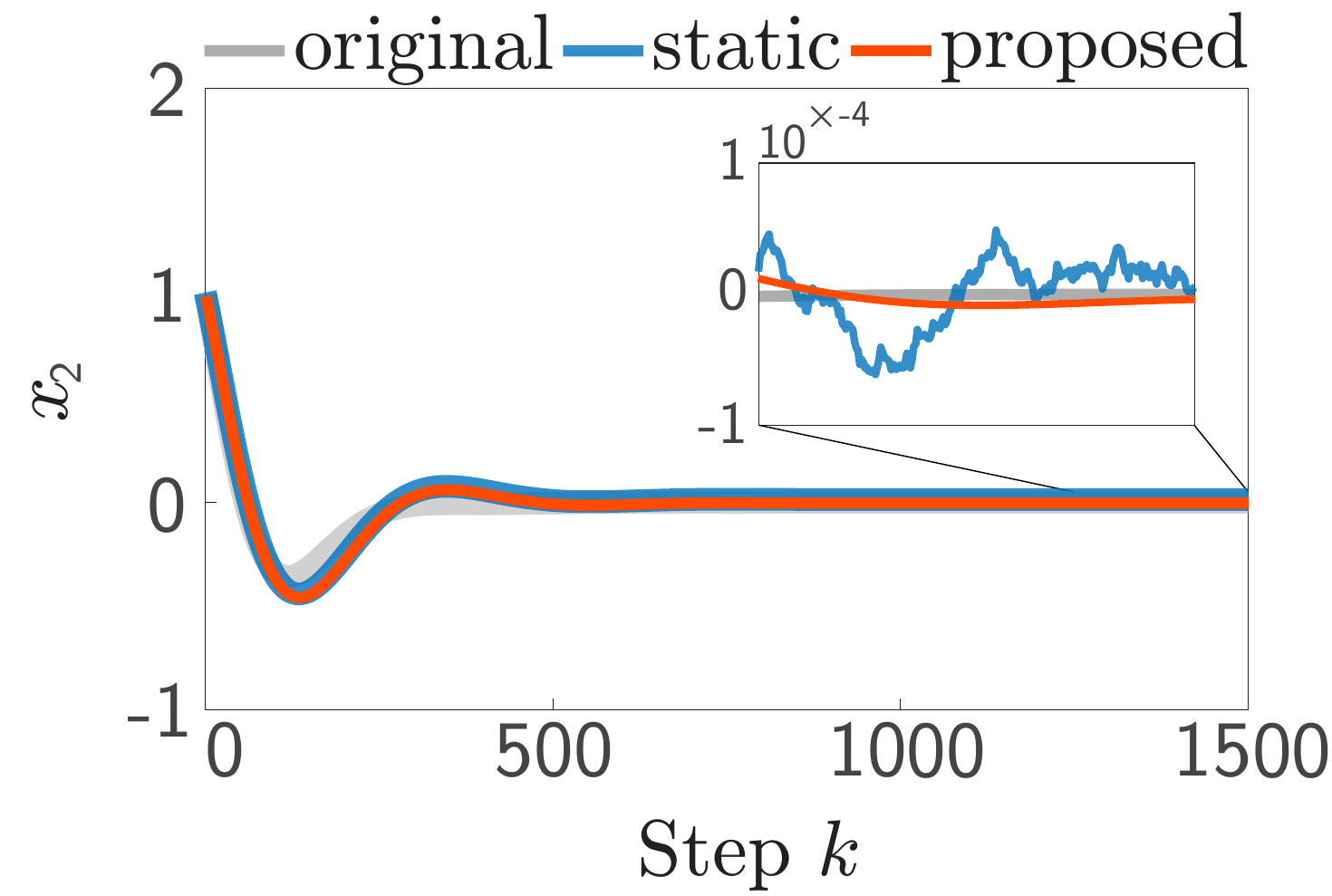}
        \caption{$x_2$}
    \end{subfigure}
    \hfill
    \begin{subfigure}{0.47\linewidth}
        \centering
        \includegraphics[width=\linewidth]{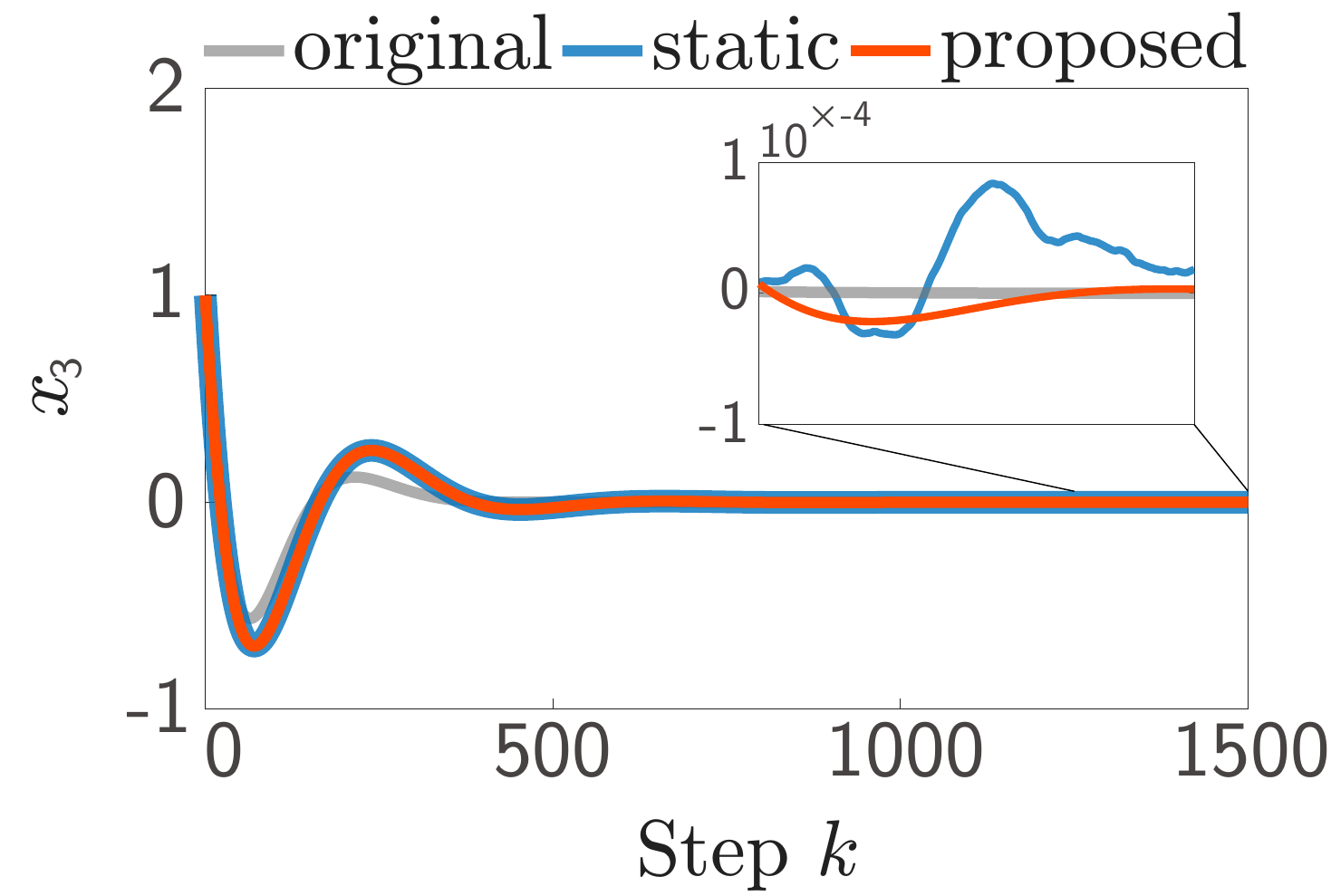}
        \caption{$x_3$}
    \end{subfigure}

    \begin{subfigure}{0.47\linewidth}
        \centering
        \includegraphics[width=\linewidth]{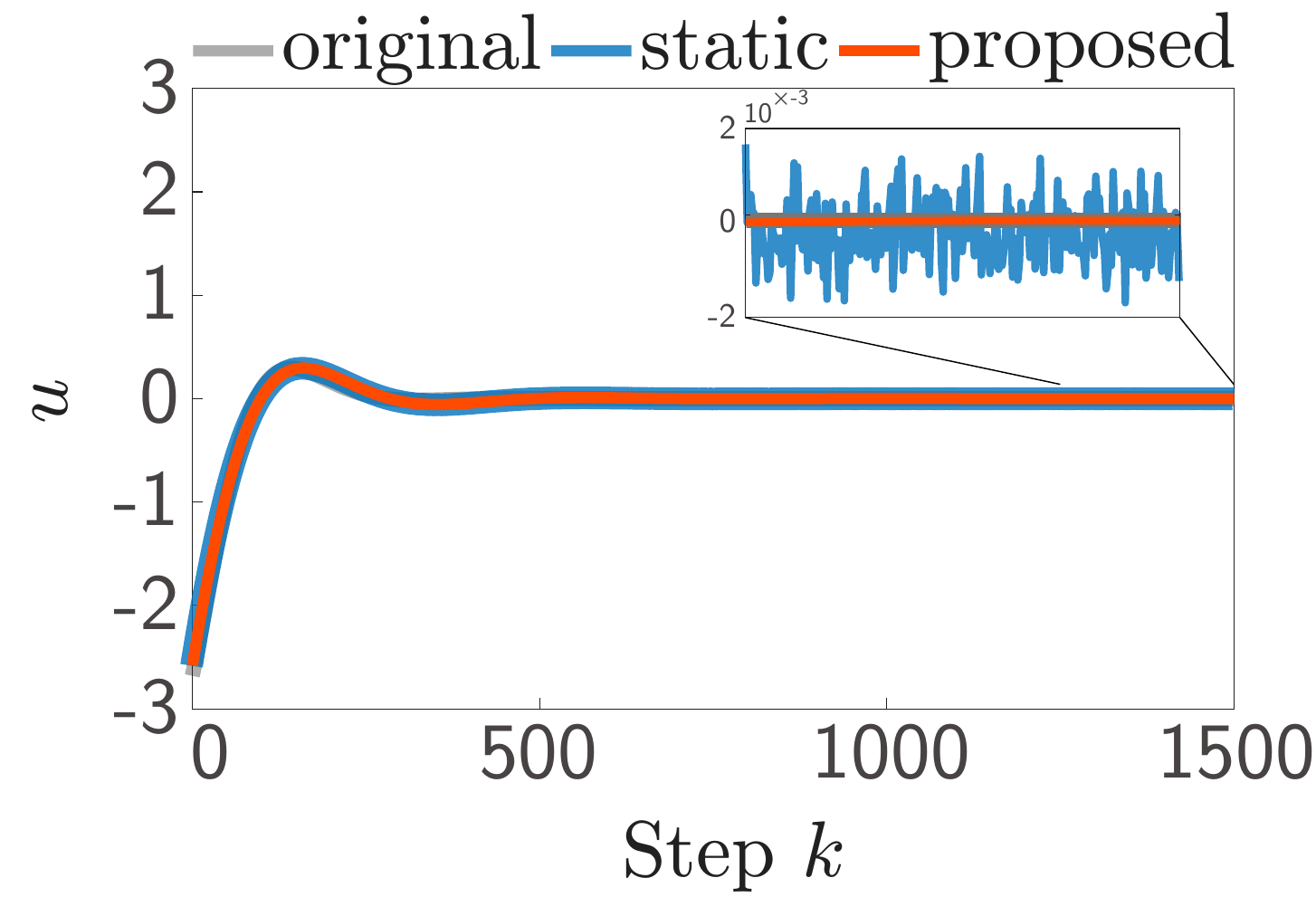}
        \caption{$u$}
    \end{subfigure}
    \hfill
    \begin{subfigure}{0.47\linewidth}
        \centering
        \includegraphics[width=1.0\linewidth]{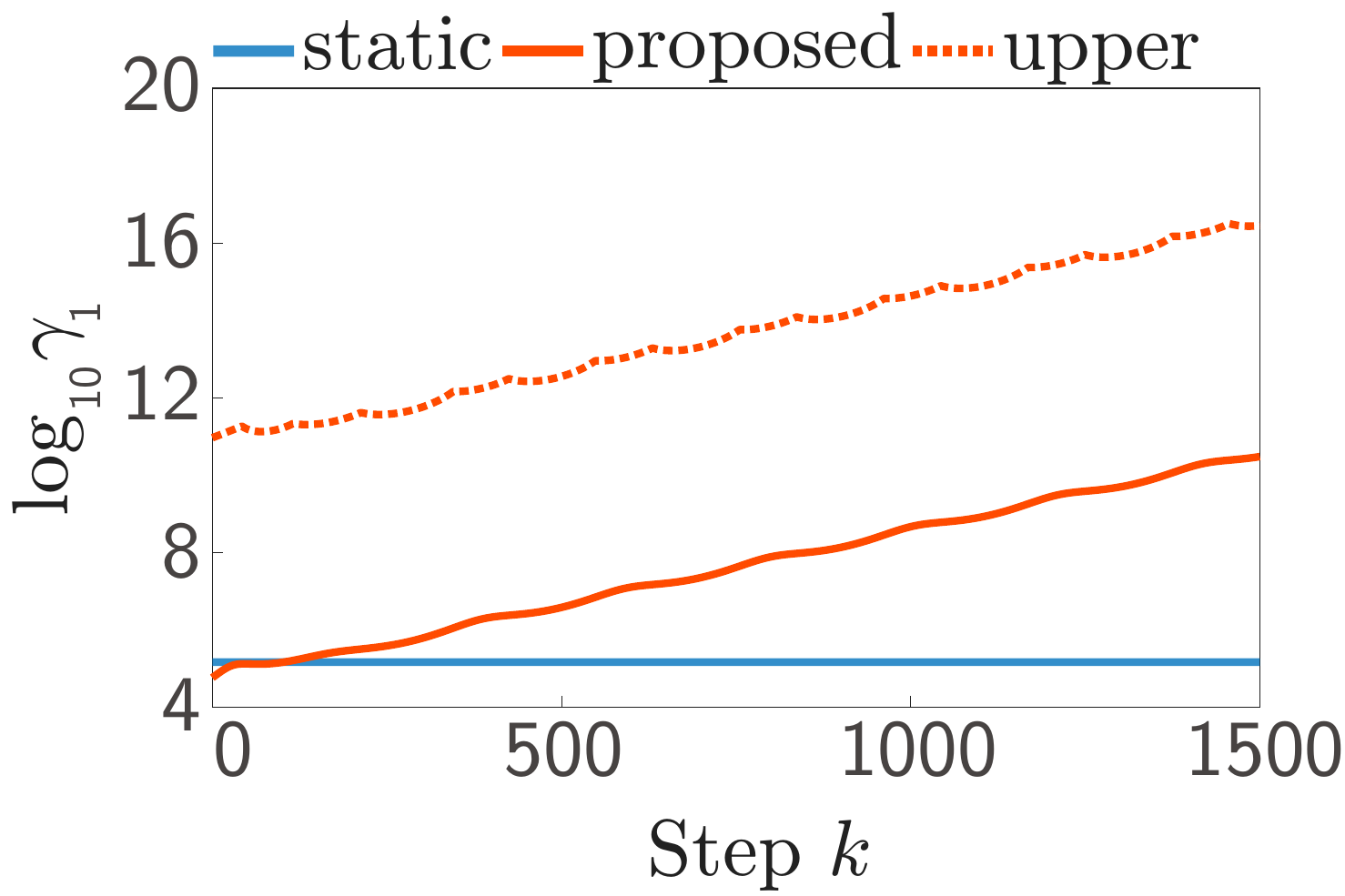}
        \caption{$\gamma_{1}$}
    \end{subfigure}

    \caption{Comparison between the signals of the original, encrypted control system with the static and time-varying encoders from \textbf{Theorem~\ref{thm:results}}, and those with static encoders alone.}
    \label{fig:dyn}
\end{figure}

\begin{figure}[tb]
    \centering
    \begin{subfigure}{0.47\linewidth}
        \centering
        \includegraphics[width=0.95\linewidth,keepaspectratio]{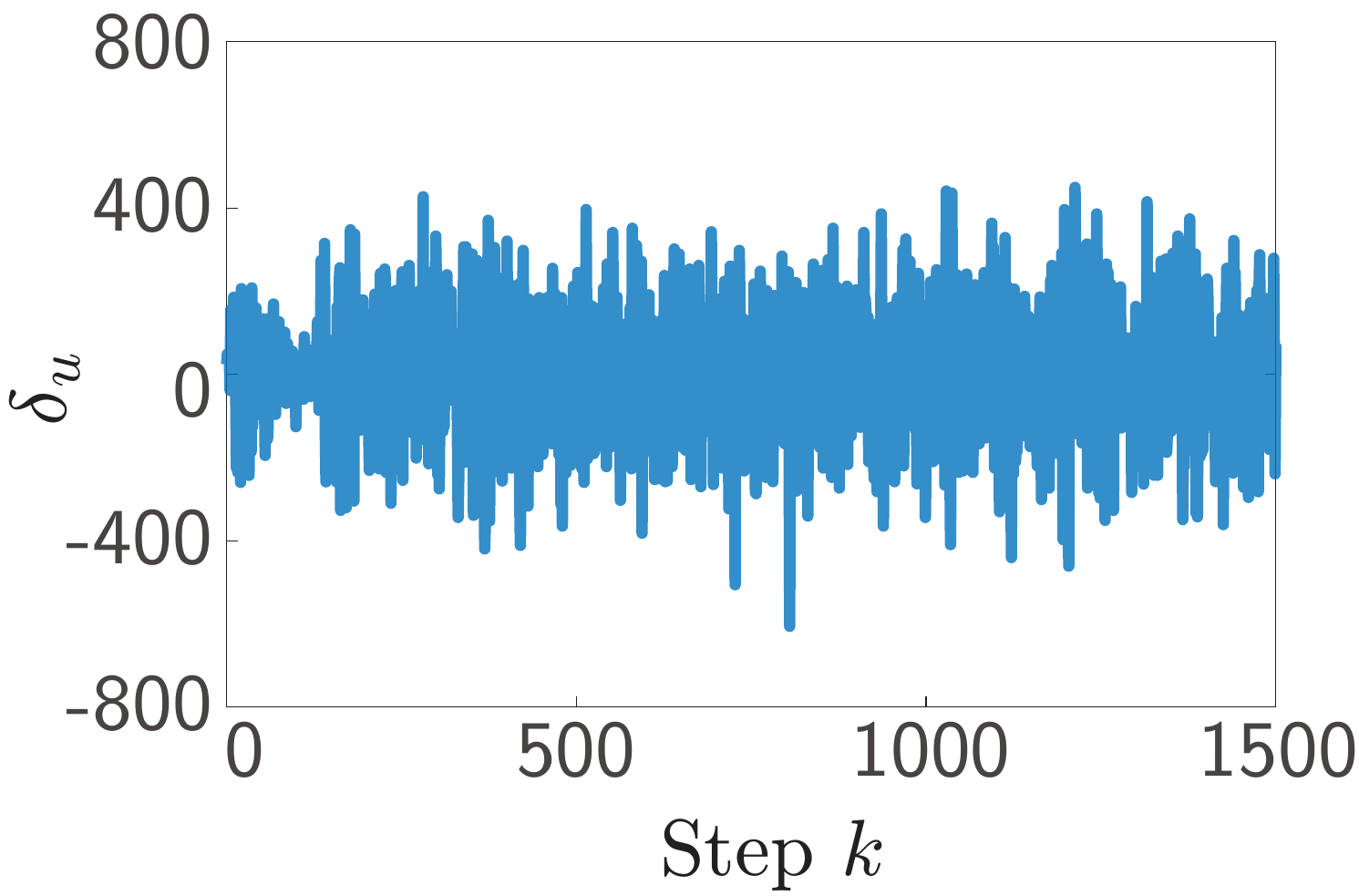}
        \caption{$\delta_u$}
    \end{subfigure}
    \hfill
    \begin{subfigure}{0.47\linewidth}
        \centering
        \includegraphics[width=\linewidth]{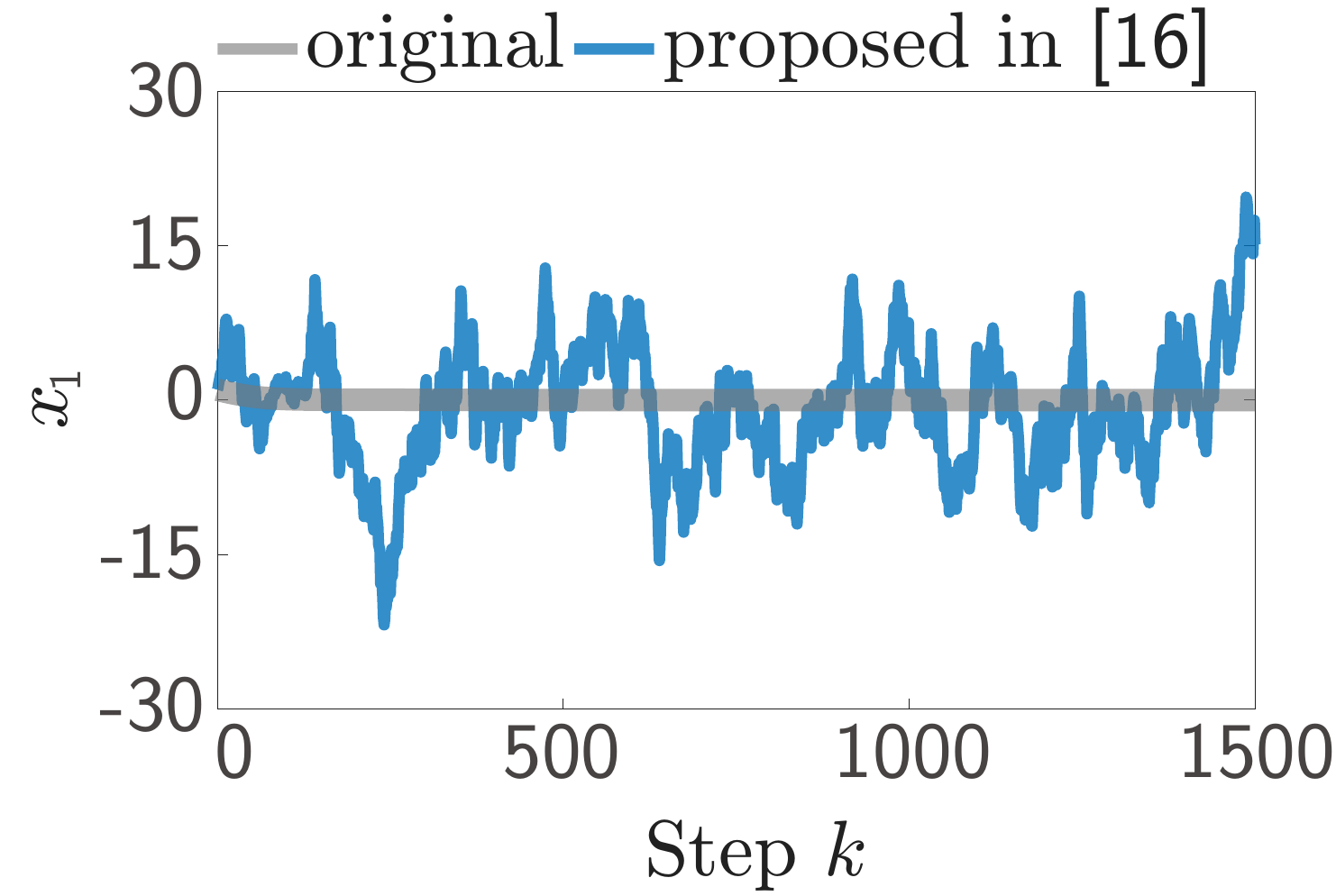}
        \caption{$x_1$}
    \end{subfigure}
    
    \begin{subfigure}{0.47\linewidth}
        \centering
        \includegraphics[width=\linewidth]{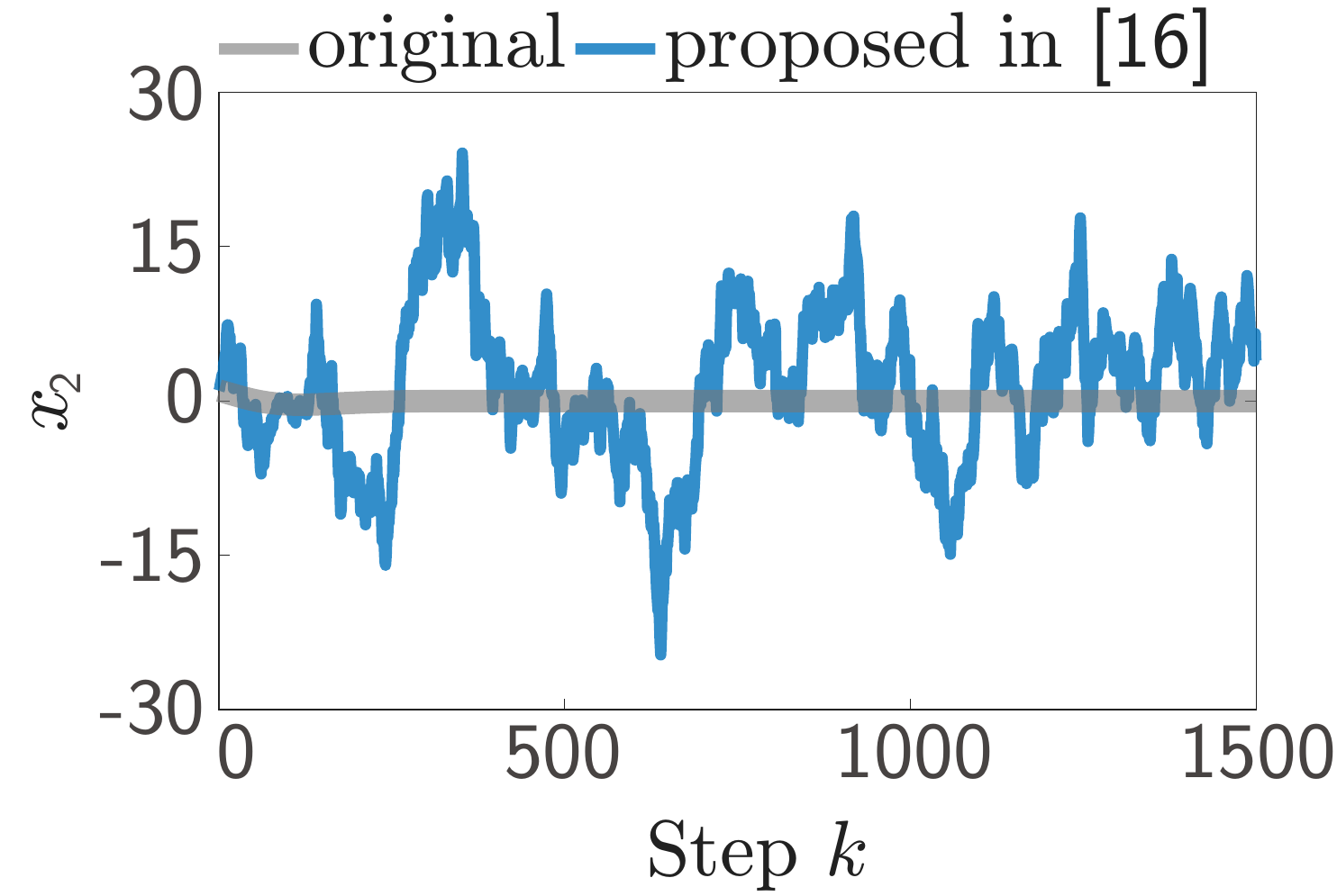}
        \caption{$x_2$}
    \end{subfigure}
    \hfill
    \begin{subfigure}{0.47\linewidth}
        \centering
        \includegraphics[width=\linewidth]{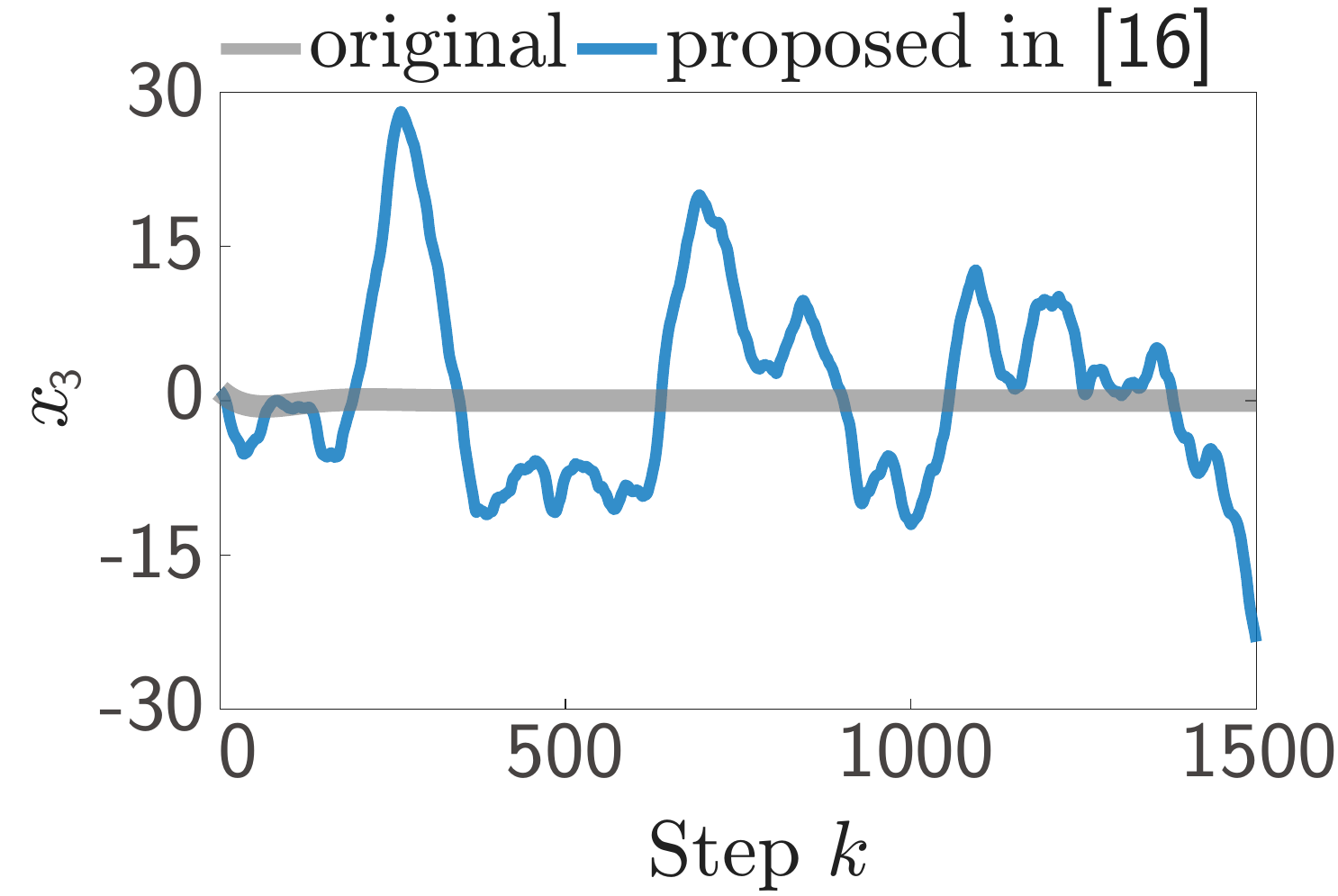}
        \caption{$x_3$}
    \end{subfigure}
    
    \begin{subfigure}{0.47\linewidth}
        \centering
        \includegraphics[width=\linewidth]{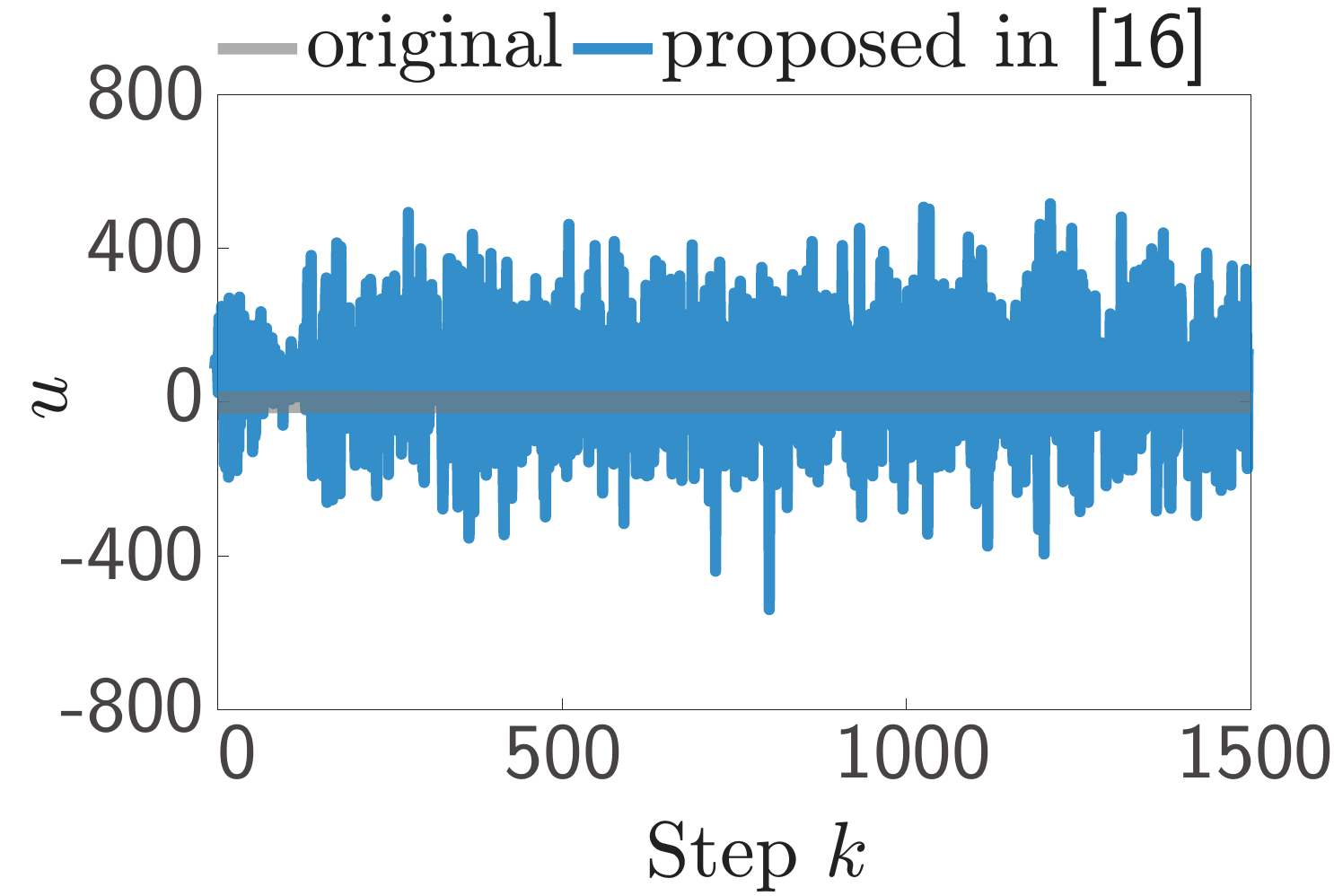}
        \caption{$u$}
    \end{subfigure}
    \hfill
    \begin{subfigure}{0.47\linewidth}
        \centering
        \includegraphics[width= 0.95\linewidth,keepaspectratio]{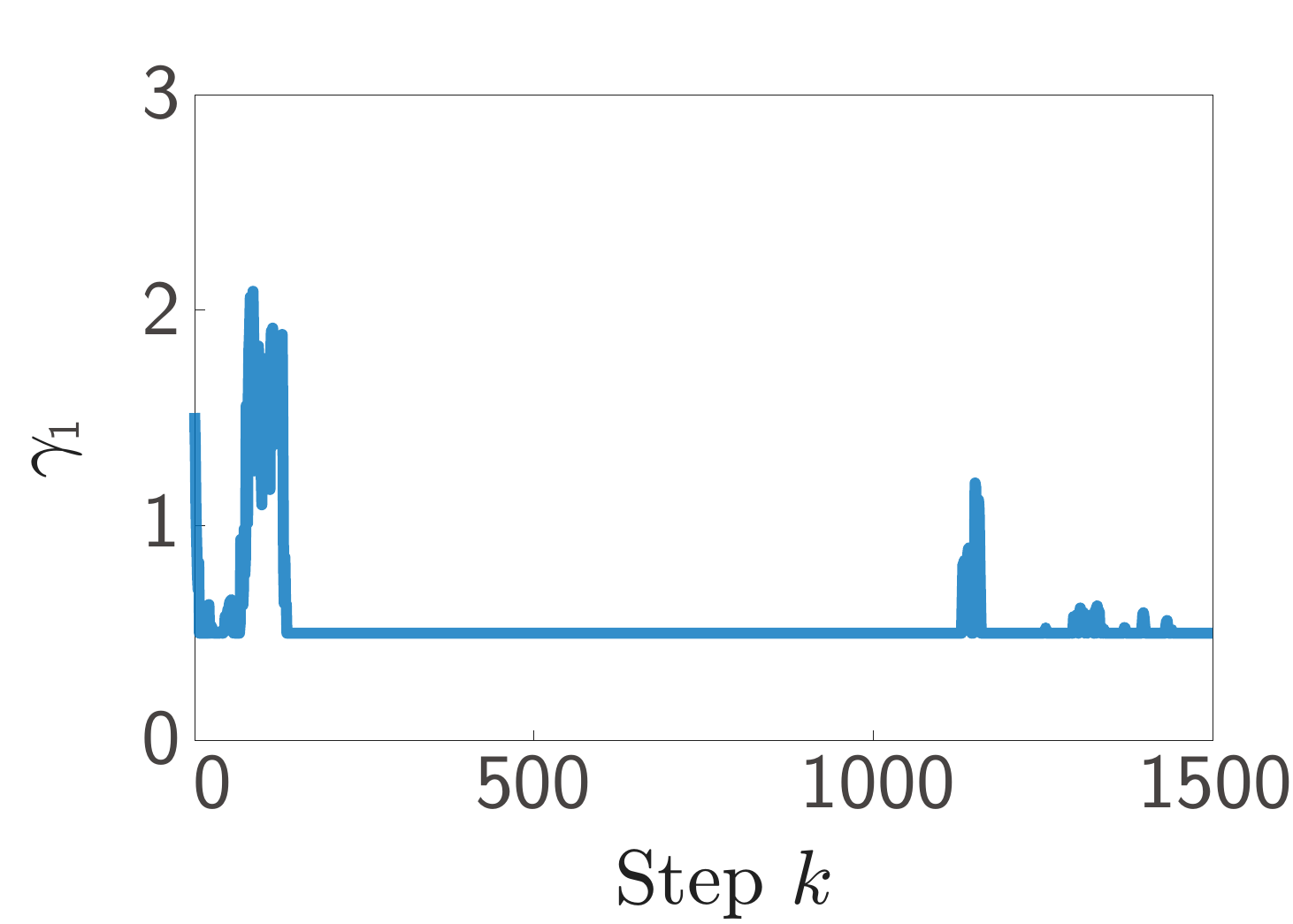}
        \caption{$\gamma_{1}$}
    \end{subfigure}
    \caption{Simulation results of the encrypted control system with static and time-varying encoders from~\cite{Teranishi20}.}
    \label{fig:dyn_wrong}
    \vspace{-3ex}
\end{figure}

\textrevise{
We measured the average per-step processing time over $10000$ control steps on the same platform.
The mean computation times were $(9.0 \pm 5.0)\times 10^{-4}$\,ms for encoding, $(3.7 \pm 0.9)\times 10^{-3}$\,ms for encryption, $(3.60 \pm 0.03)\times 10^{-1}$\,ms for homomorphic operations, $(1.0 \pm 0.5)\times 10^{-3}$\,ms for decryption, approximately zero for decoding and key update, $(3.85 \pm 0.05)\times 10^{-1}$\,ms for ciphertext update, giving a total time of $(7.51 \pm 0.07)\times 10^{-1}$\,ms per control step.
These results suggest that the proposed dynamic-key LWE-based encrypted controller has the potential for real-time implementation in practical control systems with millisecond-scale sampling periods.
}

\section{Conclusion}
In this study, we analyzed reliable dynamic-key LWE-based encrypted state-feedback control systems with time-varying encoders and decoders. 
We derived conditions on the corresponding quantization gains that guarantee asymptotic stability of the closed-loop system, and characterized the admissible range for avoiding overflow that ensures numerical safety.
A numerical example demonstrated that the proposed time-varying encoders and decoders achieve both asymptotic stability and numerical safety.

Future work will investigate the practical implementation of the proposed framework in real systems. 
In addition, we will extend the proposed framework to more general control system settings.

\section*{Acknowledgment}
This work was supported in part by the Japan Society for the Promotion of Science (JSPS) KAKENHI Grant Number JP26K00966.

\bibliographystyle{IEEEtran}
\bibliography{references}

\end{document}